\documentclass[runningheads]{llncs}
\usepackage{color, colortbl}
\usepackage{stmaryrd}
\usepackage{amsmath,amssymb}
\usepackage{bm}
\usepackage[all]{xy} 
\usepackage{xfrac}
\usepackage{tikz,tikzit}

\tikzstyle{red dot}=[fill=red, shape=circle, scale=0.4]
\tikzstyle{blue dot}=[fill=blue, shape=circle, scale=0.4]
\tikzstyle{dark red dot}=[fill={rgb,255: red,191; green,0; blue,64}, shape=circle, scale=0.3]
\tikzstyle{medium box}=[fill=white, draw=black, shape=rectangle, minimum width=0.3cm, minimum height=0.5cm]
\tikzstyle{s flat}=[fill=white, draw=black, shape=rectangle, minimum width=8mm, minimum height=5mm]
\tikzstyle{black dot}=[fill=black, draw=black, shape=circle, scale=0.4]
\tikzstyle{empty dot}=[fill=none, draw=black, shape=circle, scale=0.4]
\tikzstyle{l flat}=[fill=white, draw=black, shape=rectangle, minimum width=1.8cm, minimum height=0.3cm]
\tikzstyle{s rect}=[fill=white, draw=black, shape=rectangle, minimum width=0.1cm, minimum height=0.1cm]
\tikzstyle{s vert}=[fill=white, draw=black, shape=rectangle, minimum width=5mm, minimum height=8mm]
\tikzstyle{m vert}=[fill=white, draw=black, shape=rectangle, minimum width=5mm, minimum height=12mm]
\tikzstyle{m flat}=[fill=white, draw=black, shape=rectangle, minimum width=6mm, minimum height=5mm]
\tikzstyle{mm flat}=[fill=white, draw=black, shape=rectangle, minimum height=5mm, minimum width=10mm]
\tikzstyle{mmm flat}=[fill=white, draw=black, shape=rectangle, minimum height=5mm, minimum width=12mm]
\tikzstyle{grey dot}=[fill={rgb,255: red,191; green,191; blue,191}, draw={rgb,255: red,191; green,191; blue,191}, shape=circle, scale=0.4]
\tikzstyle{mm vert}=[fill=white, draw=black, shape=rectangle, minimum width=5mm, minimum height=14mm]
\tikzstyle{15mm vert}=[fill=white, draw=black, shape=rectangle, minimum height=15mm, minimum width=5mm]
\tikzstyle{16mm vert}=[fill=white, draw=black, shape=rectangle, minimum width=5mm, minimum height=16mm]
\tikzstyle{18mm vert}=[fill=white, draw=black, shape=rectangle, minimum width=5mm, minimum height=18mm]
\tikzstyle{20mm vert}=[fill=white, draw=black, shape=rectangle, minimum width=5mm, minimum height=20mm]
\tikzstyle{10mm vert}=[fill=white, draw=black, shape=rectangle, minimum height=10mm, minimum width=5mm]

\tikzstyle{dashes}=[-, dashed]
\tikzstyle{right arrow}=[->]
\tikzstyle{left arrow}=[<-]
\tikzstyle{grey fill}=[-, fill={rgb,255: red,191; green,191; blue,191}, draw={rgb,255: red,191; green,191; blue,191}]
\tikzstyle{blue fill}=[-, fill=cyan, draw=cyan]
\tikzstyle{yellow fill}=[-, fill=yellow, draw=yellow]
\tikzstyle{green fill}=[-, fill=green, draw=green]
\tikzstyle{red wire}=[-, draw=red]
\tikzstyle{blue wire}=[-, draw=blue]
\tikzstyle{grey wire}=[-, draw={rgb,255: red,128; green,128; blue,128}]
\tikzstyle{green wire}=[-, draw=green]
\tikzstyle{yellow wire}=[-, draw=yellow]
\tikzstyle{black fill}=[-, fill=black]

\usepackage{algpseudocode}
\usepackage{mathpartir}
\usepackage{xspace}
\usepackage{hyperref}
\usepackage{multicol}
\usetikzlibrary{patterns}
\usepackage{wrapfig}
\usepackage{graphicx}
\usepackage{cleveref}
\usepackage{subcaption}
\usepackage[inline]{enumitem}
\usepackage{mathrsfs}
\usepackage{mathbbol}
\usepackage{lineno}
\usepackage{thm-restate}
\usepackage{mathdots}
\usepackage{caption}
\usepackage{subcaption}
\usepackage{cutwin}
\usepackage{etoolbox}

\AtEndEnvironment{proof}{\hfill {$\square$}}

\makeatletter
\newcommand*{\inlineequation}[2][]{%
  \begingroup
    \refstepcounter{equation}%
    \ifx\\#1\\%
    \else
      \label{#1}%
    \fi
    \relpenalty=10000 %
    \binoppenalty=10000 %
    \ensuremath{%
      #2%
    }%
    ~\@eqnnum
  \endgroup
}
\makeatother

\newcommand{\defeq}{\mathrel{\mathop:}=}

\def\ci{\perp\!\!\!\perp}

    \newcommand{\markovCI}[3]{#1 \!\perp_{\!\mathrm{M}}\! #2 \!\mid\! #3}
    
    \newcommand{\procCI}[3]{#1 \!\perp_{\!\mathrm{S}}\! #2 \!\mid\! #3}
    
    \newcommand{\extSupsetCI}[3]{#1 \!\perp_{\!\mathrm{S^+}}\! #2 \!\mid\! #3}

    \newcommand{\dibiCI}[3]{#1 \!\ci_{\!\mathrm{L}}\! #2 \!\mid\! #3}
    
    \newcommand{\displayCI}[3]{#1 \!\perp\! #2 \!\mid\! #3}

\newcommand{\oran}[1]{\textcolor{orange}{#1}}
\newcommand{\purp}[1]{\textcolor{violet}{#1}}

\newcommand{\jh}[1]{\textcolor{green}{[Justin: {#1}]}}
\newcommand{\jb}[1]{\textcolor{brown}{[Jialu: {#1}]}}
\newcommand{\rv}[1]{\textcolor{blue}{[Reviewer: {#1}]}}
\newcommand{\sd}[1]{\ifdraft{\color{red} [Simon: {#1}]}\fi}
\newcommand{\as}[1]{\ifdraft{\color{pink} [A: {#1}]}\fi}
\newcommand{\tao}[1]{\textcolor{purple}{[Tao: {#1}]}}
\newcommand{\fz}[1]{\textcolor{blue} {[Fabio: {#1}]}}
\newcommand{\todo}[1]{\text{\textcolor{red}{TODO: {#1}}}}
\newcommand{\hush}[1]{}
\newcommand{\quieteveryone}{%
	\let\sd\hush%
	\let\jh\hush%
	\let\jb\hush%
	\let\as\hush%
	\let\tao\hush%
	\let\fz\hush%
    \let\todo\hush%
    \let\rv\hush%
}

\usepackage{mathtools}
\makeatletter
\def\@acmplainindent{0pt}
\def\@acmdefinitionindent{0pt}
\def\@proofindent{\noindent}
\makeatother

\newcommand{\Var}{\setVar}

\newcommand{\Val}{\ensuremath{\mathbf{Val}}}

\newcommand{\Bool}{\{ 0, 1 \}}
\newcommand{\dbind}{\bind}

\newcommand{\kp}[2]{\ensuremath{{#1} \mathrel{\triangleright} [{#2}]}}

\newcommand{\bind}{\textsf{bind}}

\newcommand{\Kl}{\mathcal{K}\mkern-2mu\ell}

\newcommand{\stPlus}{\oplus}
\newcommand{\stThen}{\mathbin{\odot}}

\newcommand{\stSmaller}{\sqsubseteq}
    \newcommand{\stLeq}{\sqsubseteq}
    \newcommand{\stGeq}{\sqsupseteq}

\newcommand{\frameA}{\mathcal{A}}
\newcommand{\unitSet}{E}
\newcommand{\powset}{\mathcal{P}}

\newcommand{\ipowset}{\powset_{\!i}}
    
\newcommand{\onlyif}{\Longrightarrow}
\newcommand{\catC}{\mathbb{C}}
\newcommand{\catD}{\mathbb{D}}
\newcommand{\catX}{\mathbb{X}}

\newcommand{\obX}{\mathsf{X}}
    \newcommand{\obY}{\mathsf{Y}}
    \newcommand{\obZ}{\mathsf{Z}}
    \newcommand{\obU}{\mathsf{U}}
    \newcommand{\obV}{\mathsf{V}}
    \newcommand{\obW}{\mathsf{W}}
    \newcommand{\obC}{\mathsf{C}}
    \newcommand{\obD}{\mathsf{D}}
    \newcommand{\obA}{\mathsf{A}}
    \newcommand{\obB}{\mathsf{B}}
    \newcommand{\obE}{\mathsf{E}}

\newcommand{\tensor}{\otimes}
\newcommand{\tenUnit}{\mathsf{I}}
\newcommand{\lr}[1]{\langle #1 \rangle}
\newcommand{\lrangle}[1]{\langle #1 \rangle}
\newcommand{\copyMor}{\mathsf{copy}}
\newcommand{\copier}{\input{copier.tikz}}
\newcommand{\discardMor}{\mathsf{del}}
\newcommand{\discarder}{\begin{tikzpicture}
	\begin{pgfonlayer}{nodelayer}
		\node [style=none] (0) at (-0.5, 0) {};
		\node [style=black dot] (1) at (0.25, 0) {};
	\end{pgfonlayer}
	\begin{pgfonlayer}{edgelayer}
		\draw (0.center) to (1);
	\end{pgfonlayer}
\end{tikzpicture}
}

\newcommand{\swap}{\input{swap.tikz}}

\newcommand{\ob}[1]{\mathbf{ob}(#1)}
    \newcommand{\mor}[1]{\mathbf{mor}(#1)}
\newcommand{\id}{\mathit{id}}
    \newcommand{\idMor}{\mathit{id}}
\newcommand{\after}{\circ}
\newcommand{\setVar}{\mathrm{Var}}

\newcommand{\varLeq}{\preceq}
    \newcommand{\varLessThan}{\prec}
\newcommand{\funcT}{\mathcal{T}}

\newcommand{\catSet}{\mathsf{Set}}

\newcommand{\setOne}{\mathbf{1}}

\newcommand{\setVal}{\mathrm{Val}}
\newcommand{\decomp}[1]{\lr{#1}}

\newcommand{\domain}[1]{\mathit{dom}(#1)}
\newcommand{\codom}[1]{\mathit{cod}(#1)}

\newcommand{\emptyList}{[~]}

\newcommand{\ket}[1]{|#1\rangle}

\newcommand{\dist}{\mathcal{D}}

\newcommand{\empset}{\varnothing}

\newcommand{\proj}{\pi}

\newcommand{\memProj}[2]{{#1}^{#2}}

\newcommand{\mand}{*}
\newcommand{\mimp}{\mathrel{-\mkern-6mu*}}
\newcommand{\biThen}{\mathbin{\fatsemi}}
\newcommand{\biTo}{\multimap}
    \newcommand{\biOt}{\mathrel{\reflectbox{$\multimap$}}}

\newcommand{\klto}{\rightarrowtriangle}

\newcommand{\memoryL}{\bm{\ell}}
\newcommand{\memoryM}{\mathbf{m}}
\newcommand{\memoryN}{\mathbf{n}}
\newcommand{\fsubseteq}{\subseteq_{\!f}\!}
\newcommand{\monadUnit}[2]{\eta^{#1}_{#2}}
    \newcommand{\monadMult}[2]{\mu^{#1}_{#2}}

\newcommand{\setAtomProp}{\mathcal{AP}}
\newcommand{\smallDIBI}{\ensuremath{\text{DIBI}_{\{\land, \mand, \biThen\}}}}
\newcommand{\valV}{\mathcal{V}}
    \newcommand{\natVal}{\mathcal{V}_{\mathrm{nat}}} 
\newcommand{\satisfy}[1]{\vDash_{#1}}

\newcommand{\memSpace}[2]{\mathbf{M}[#1; #2]}
    \newcommand{\memSpaceS}[1]{\mathbf{M}[#1]}

\newcommand{\Values}{\mathrm{Val}}

    \newcommand{\setToList}[1]{\llbracket #1 \rrbracket}

\newcommand{\xleq}[1]{}
\newcommand{\eqWhenDefined}{\doteq}

\newcommand{\catWithVarName}[2]{#1[#2]}
\newcommand{\obChoice}{\theta}
\newcommand{\catSynVar}{\mathbb{SynVar}}
\newcommand{\catKernel}[1]{\mathbb{Ker}(#1)}
\newcommand{\catMemProb}{\mathbb{MemPr}}

\newcommand{\tenSwap}{\sigma}

\newcommand{\conProbKern}{\mathrm{ProbKer}}

\newcommand{\operateVar}{\star}
\newcommand{\listL}{L}
\newcommand{\listK}{K}

    \newcommand{\frameSynKernel}{\mathbf{SynFr}}
    \newcommand{\setSemKernel}[1]{\mathrm{Ker}(#1)}
    \newcommand{\frameSemKernel}[1]{\mathbf{Fr}(#1)}
    
    \newcommand{\probFrame}{\mathbf{PrFr}}
    \newcommand{\probFrameBasedOn}[1]{\mathbf{PrFr}[#1]}
\newcommand{\listx}{\bar{x}}
\newcommand{\listy}{\bar{y}}

\newcommand{\catProbKern}{\mathbb{PrKern}}

\newcommand{\catDibiFrame}{\mathbb{DibiFr}}

\newcommand{\borel}{\mathcal{B}}
\newcommand{\real}{\mathbb{R}}
\newcommand{\catMeas}{\mathbb{Meas}}

\newcommand{\sigalg}[1]{\Sigma_{\!{#1}}}
\newcommand{\monadGiry}{\mathcal{G}}
\newcommand{\monadRadon}{\mathcal{R}}
\newcommand{\ev}{\text{ev}}
\newcommand{\dirac}[1]{\delta_{#1}}
\newcommand{\catStoch}{\mathbb{Stoch}}
\newcommand{\catHaus}{\mathbb{CHaus}}
\newcommand{\catGauss}{\mathbb{Gauss}}
\newcommand{\catBorelStoch}{\mathbb{BorelStoch}}

\newcommand{\uniform}{\textsc{Unif}}

\newcommand{\trivList}{\langle \rangle}
\newcommand{\normal}{\mathcal{N}}
\newcommand{\cov}{\sigma^2}
\newcommand{\EE}{\mathbb{E}}
\newcommand{\initialMap}{!}

\newcommand{\isDefined}{\!\Downarrow\!}
\newcommand{\eqIfDefined}{\doteq}
\newcommand{\support}[1]{\mathrm{supp}(#1)}
\newcommand{\memComb}{\uplus}

\begin{document}
\title{A Categorical Approach to DIBI Models}
%
%
\author{Tao Gu\inst{1} \and
Jialu Bao\inst{2} \and
Justin Hsu\inst{2} \and Alexandra Silva\inst{2} \and Fabio Zanasi\inst{1,3} }
\authorrunning{Gu, Bao, Hsu, Silva, and Zanasi}
%
\institute{University College London, UK
\and
Cornell University, NY, USA
\and University of Bologna, Italy\\
}
\maketitle              
\begin{abstract}
The logic of Dependence and Independence Bunched Implications (DIBI) is a logic to reason about conditional independence (CI); for instance, DIBI formulas can characterise CI in probability distributions and relational databases, using the probabilistic and relational DIBI models, respectively.
Despite the similarity of the probabilistic and relational models, a uniform, more abstract account remains unsolved. The laborious case-by-case verification of the frame conditions required for constructing new models also calls for such a treatment.
In this paper, we develop an abstract framework for systematically constructing DIBI models, using category theory as the unifying mathematical language. In particular, we use string diagrams -- a graphical presentation of monoidal categories -- to give a uniform definition of the parallel composition and subkernel relation in DIBI models.
Our approach not only generalises known models, but also yields new models of interest and reduces properties of DIBI models to structures in the underlying categories.
Furthermore, our categorical framework enables a logical notion of CI, in terms of the satisfaction of specific DIBI formulas. We compare it with string diagrammatic approaches to CI and show that it is an extension of string diagrammatic CI under reasonable conditions.
\end{abstract}
\section{Introduction}
\label{sec:intro}
Conditional independence (CI) is a fundamental concept across various research areas, including programming languages~\cite{reynolds2002separation,barthe2019probabilistic,li2023lilac}, statistics~\cite{dawid1979conditional}, and database theory~\cite{aho1979theory}, among others. Although specific definitions may vary, the core idea remains straightforward: events $A$ and $B$ are   `independent' when information about one event does not convey information about the other. Furthermore, events $A$ and $B$ are `conditionally independent' given event $C$ if, with knowledge of event $C$, events $A$ and $B$ become independent. Albeit intuitive, reasoning about conditional independence is an intricate task, leading to extensive research efforts aimed at formalising such reasoning~\cite{pearl2009causality,glymour2019review}.

For probabilistic programs, an extension of standard programs with constructs to sample from distributions, formal methods for (conditional) independence have emerged as powerful tools for program verification.
For instance, Barthe et. al.~\cite{barthe2019probabilistic} introduced Probabilistic Separation Logic (PSL) and applied it to formalise several cryptography protocols, where independence of  variables guarantees no leakage of information and thus security of the algorithms.
A follow-up work from Bao et al.~\cite{bao2021bunched}
proposed the logic of \emph{Dependence and Independence Bunched Implications} (DIBI),
which enhances PSL with the ability to reason about \emph{conditional} independence. Syntactically, DIBI extends the logic of Bunched Implications (BI)~\cite{o1999logic,pym2013semantics}, which is the assertion logic underpinning Separation Logic (SL)~\cite{reynolds2002separation} and PSL, with a non-commutative conjunction $\biThen$ and its adjoints.
Semantically, as in BI, the separating conjunction $\mand$ is interpreted through a partial operation $\stPlus$ on states, regarded as the parallel composition. In addition, they define a sequential composition $\stThen$ to interpret $P \biThen Q$.
Informally, $P \mand Q$ says that $P$ and $Q$ hold in states that can be separated, and $P \biThen Q$ expresses a possible dependency of $Q$ on $P$.  \Cref{sec:DIBI} will review the logic in more details.

Bao et. al \cite{bao2021bunched} introduced two kinds of semantic models for DIBI logic -- probabilistic and relational.
The probabilistic DIBI models
are used to reason about CI of variables in discrete probabilistic computation.
The relational DIBI models
are designed to express join dependency -- a notion of conditional independence in database theory and relational algebra -- between variables in relational databases.
These two models share many similarities, and the conditions to verify they are models are repetitive. This led Bao et. al.~\cite{bao2021bunched} to conjecture a categorical approach to construct abstract DIBI models that induce these concrete models as instances. This would greatly facilitate the construction of new models.

To solve this conjecture, we start with the observation that, in both the probabilistic and relational DIBI models, the states
resemble \emph{Markov kernels}: they are maps from
input elements to distributions/powersets over output elements. Specifically, the input/output elements are value assignments on a finite set of variables, as an abstraction of program memories or database entries.
Such DIBI states can be identified categorically as morphisms in the Kleisli categories associated to the discrete distribution monad (\cref{def:dist-monad}, for the probabilistic model) or the nonempty powerset monad (\cref{def:powerset-monad}, for the relational model).
However, giving a categorical definition for the parallel compositions $\stPlus$ is difficult.
The previous work~\cite{bao2021bunched} gives~\Cref{fig:pic-parallel} as a pictorial intuition for the parallel composition.
The states are drawn as trapezoids, with the short and long vertical sides representing the input and output domains, respectively.
There, given a blue map $f_1$ and a red map $f_2$, their parallel composition $f_1 \oplus f_2$ takes as input the union of their inputs. Then, each $f_i$ takes its counterpart in the combined input domain and generates an output. Finally these two outputs are combined to be the output of $f_1 \oplus f_2$.
This parallel composition \emph{is partial} because the combination of their outputs is allowed only when the variables overlap in particular ways.
This creates a challenge to capture DIBI models categorically because, in a categorical setting, the domains and codomains of DIBI states are objects, and it is not obvious how to define the overlap of objects used in the parallel composition.
\begin{figure}[bt]
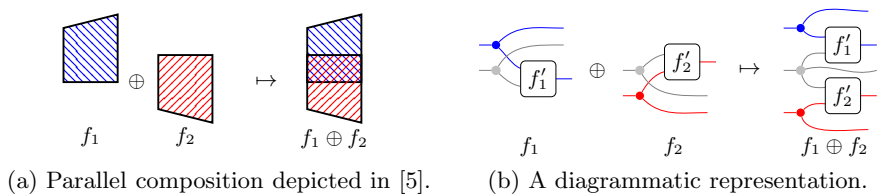

    \begin{subfigure}[b]{0.5\textwidth}
		\begin{center}
		\scalebox{.9}{
			\input{figures/parallel-comp-intuition.tikz}}
		\end{center}
        \vspace*{-.4cm}
        \caption{ Parallel composition depicted in~\cite{bao2021bunched}.}
        \label{fig:pic-parallel}
        \end{subfigure}%
        \hfill
        \begin{subfigure}[b]{0.5\textwidth}
        \begin{center}
            \scalebox{.9}{\input{compose-intuition-1.tikz}}
        \end{center}
        \vspace*{-.4cm}
        \caption{A diagrammatic representation.}
        \label{fig:string-diag-para}
    	\end{subfigure}
     \caption{Intuition for parallel composition.}
     \vspace*{-.6cm}
\end{figure}

Our solution stems from a formalisation of this graphical intuition through \emph{string diagrams}, a pictorial formalism for monoidal categories. String diagrams are widely adopted as intuitive yet mathematically rigorous reasoning tools across different areas of science, see~\cite{piedeleuzanasi23} for an overview.
We formalise the trapezoids intuition in \Cref{fig:pic-parallel} into string diagrams as in \Cref{fig:string-diag-para}.
The maps previously embodied as trapezoids now have a fork shape, with some branches being straight lines and some other branches going through boxes.
The boxes represent arbitrary morphisms in the underlying category, and the straight lines represent the identity morphisms.  
Whereas composition of two DIBI states were hand-waved as two trapezoids tiled together in~\Cref{fig:pic-parallel}, with string diagram we can define it precisely:
the overlap of the two trapezoids is witnessed by the grey wires, and
the composition joins two diagrams side-by-side with the grey wires shared.
We will show in~\Cref{sec:Markov-DIBI-model} that this string diagram representation yields DIBI models in any category with enough structure to interpret $\input{copier.tikz}$, namely, Markov categories~\cite{cho2019disintegration,fritz2020markov}.
Furthermore, in \Cref{sec:applications}, we will derive several concrete DIBI models as instances.

Additionally, our framework enables a comparison between different characterisations  of conditional independence (CI).
The previous work~\cite{bao2021bunched} expresses probabilistic or relational CI in terms of satisfaction of some DIBI formulas.
Since we can construct categorical DIBI models based on any Markov categories, we define a logical notion of CI for morphisms in Markov categories as  satisfaction of those DIBI formulas.
In \Cref{sec:CI}, we investigate the relationship between our `logical' CI and various CI notions based on categorical structures from literature in synthetic statistics~\cite{cho2019disintegration,fritz2020markov}.

Throughout the paper we fix a countably infinite set of variables $\setVar$, use $x, y, z, \dots$ for elements of $\setVar$, and use $W, X, Y, \dots$ for finite subsets of $\setVar$.
\section{Category Theory Preliminaries}\label{sec:preliminaries}
Unless specified, we assume that all monoidal categories we consider are strict and write $\domain{f}$ and $\codom{f}$ for the domain and codomain of any morphism $f$.
We write $\lr{\catC, \tensor, \tenUnit}$ for a (strict) monoidal category, where $\tensor$ is the monoidal product and $\tenUnit$ the unit object of $\catC$. If is also symmetric, we write $\tenSwap_{\obA, \obB} \colon \obA \tensor \obB \to \obB \tensor \obA$ for the symmetry natural transformation indexed by objects $\obA$ and $\obB$.

As detailed for instance in~\cite{selinger2010survey,piedeleuzanasi23,fong2019invitation}, morphisms of symmetric monoidal categories have a graphical presentation as string diagrams, where sequential composition and monoidal product are depicted as concatenation and juxtaposition of diagrams, respectively: given morphisms $f \colon \obX \to \obY$, $g \colon \obY \to \obZ$, $h \colon \obU \to \obV$,
\[
    g \after f = \input{seq-comp.tikz} \qquad g \tensor h = \input{para-comp.tikz}
\]
Our convention is to read string diagrams from left to right, and tensor products from top to bottom. We will sometimes omit object labels in the diagrams when they are evident or irrelevant to the context. Symmetries are indicated with the string diagram $\input{swap.tikz}$. We call string diagrams consisting solely of combinations of $\input{swap.tikz}$s \emph{rewirings}: intuitively, they permute the order of the objects.

We will need the notion of a Markov category, which suitably generalises categories of probabilistic processes~\cite{fritz2020markov}.
First, a \emph{copy-delete category} (\emph{CD category}) is a symmetric monoidal category (SMC) $\lr{\catC, \tensor, \tenUnit}$ with `copy' $\copyMor_{\obC}$ and `delete' $\discardMor_{\obC}$ morphisms for each object $\obC$, drawn diagrammatically as $\copier$ and $\discarder$ respectively, that form a commutative comonoid:
\[
    \input{copier-ass.tikz} \qquad \input{copier-discarder.tikz} \qquad \input{copier-sym.tikz}
\]
Because of the leftmost equation above, we sometimes write a `trident' $\input{three-copies.tikz}$ for either side of it. Moreover, both $\copyMor$ and $\discardMor$ are compatible with the monoidal structure:
\[
    \input{nat-copy-1.tikz} = \input{nat-copy-2.tikz} \qquad \qquad  \begin{tikzpicture}
	\begin{pgfonlayer}{nodelayer}
		\node [style=none] (0) at (-1.25, 0) {$\scriptstyle{A \otimes B}$};
		\node [style=none] (1) at (-0.5, 0) {};
		\node [style=black dot] (2) at (0.5, 0) {};
	\end{pgfonlayer}
	\begin{pgfonlayer}{edgelayer}
		\draw (1.center) to (2);
	\end{pgfonlayer}
\end{tikzpicture}
 = \input{nat-discard-2.tikz}
\]
We say \emph{$\discardMor$ is natural} if $\input{tenUnit-terminal.tikz}$ for every morphism $f$. A \emph{Markov category} is a CD category in which $\discardMor$ is natural. A CD category $\catC$ \emph{has conditionals} if for each morphism $f \colon \obA \to \obX \tensor \obY$, there exist (not necessarily unique) morphisms $f_{\obX} \colon \obA \to \obX$ (called the \emph{marginal}) and $f_{|\obX} \colon \obX \to \obY$ (called the \emph{conditional}) such that $\scalebox{.9}{\input{def-conditional-5.tikz}} = \scalebox{.9}{\input{def-conditional-6.tikz}}$. When $\catC$ is a Markov category, such marginal $f_\obX$ is unique given $\obX$ because of the naturality of
$\discardMor$:
\[
    \input{def-conditional-7.tikz} = \input{def-conditional-8.tikz} = \input{def-conditional-9.tikz}
\]
%


\section{DIBI Logic and its Probabilistic Model}
\label{sec:DIBI}
In this section we review the logic of \emph{Dependence and Independence Bunched Implications} (DIBI). For space reasons, we focus on the discrete probabilistic model for DIBI, as introduced in~\cite{bao2021bunched}. The interested reader may refer to~\cite{bao2021bunched} for the relational model, whose construction follows similar steps.

\noindent DIBI formulas are defined inductively as follows:
\begin{equation*}
  P, Q  ::= p \in \setAtomProp
  \mid \top
  \mid I
  \mid P \land Q
  \mid P \to Q
  \mid P \mand Q
  \mid P \mimp Q
  \mid P \biThen Q
  \mid P \biTo Q
  \mid P \biOt Q
\end{equation*}
The additive conjunction $\land$ is just the standard Boolean conjunction. The multiplicative conjunction $\mand$ states that $P$ and $Q$ are independent. Both are already present in BI.  DIBI extends BI with the non-commutative conjunction $\biThen$\footnote{Not to be confused with the additive context constructor which is also denoted as $\biThen$ in the standard BI literature such as \cite{o1999logic,pym2013semantics}.}, where $P
\biThen Q$ states that $Q$ may depend on $P$. The operation $\mimp$ is adjoint to $\mand$, and $\biTo$, $\biOt$ are adjoints to $\biThen$. DIBI formulas are interpreted on
DIBI \emph{models}, each consisting of a \emph{DIBI frame}
 on a set of states $A$ and a
\emph{valuation} function $\valV \colon \setAtomProp \to \powset (A)$
that maps an atomic proposition to the set of states on which it is true.
While a BI frame is based on a partial commutative monoid~\cite{docherty2019thesis}, a DIBI frame consists of two monoids (one commutative and one not) on the same underlying set, taking care of the two non-additive conjunctions $\mand$ and $\biThen$, respectively.
%
\begin{definition}[\cite{bao2021bunched}]
\label{def:frame-cond}
A \emph{DIBI frame} is a tuple $\frameA = \langle A, \stSmaller, \stPlus,
\stThen, \unitSet \rangle$, where $A$ is a set of states, $\stSmaller$ is a preorder on $A$, $\unitSet
\subseteq A$ are units, and $\stPlus, \stThen \colon A \times A \to \powset(A)$ are partial binary operations\footnote{Note that, even though $\stThen, \stPlus$ are also partial in the models considered in~\cite{bao2021bunched}, they have type $A \times A \to \powset (A)$ in that work. This is because the authors obtain completeness of DIBI logic using a method developed by Docherty~\cite{docherty2019thesis}, which only works for the more general type. Because the operations are actually partial rather than non-deterministic, and we are not interested in completeness here, we stick to the more accurate type.}, satisfying the frame conditions in \Cref{fig:dibi-frame}.

\begin{figure}[tb]
    \centering \footnotesize
    \begin{minipage}{0.5\textwidth}
        \begin{align*}
        & a \stPlus b \eqIfDefined b \stPlus a
        \tag{$\stPlus$-\textsc{Com}}\label{frameAx:plus-com} \\
        & \exists e \in E \colon a = e \stPlus a
        \tag{$\stPlus$-\textsc{UnitExist}} \label{frameAx:plus-unit-exist} \\
        & (a \stPlus b) \stPlus c \eqIfDefined a \stPlus (b \stPlus c)
    \tag{$\stPlus$-\textsc{Assoc}}\label{frameAx:plus-ass}
    \end{align*}
    \end{minipage}%
    \begin{minipage}{0.5\textwidth}
        \begin{align*}
         & \exists e \in E \colon a = e \stThen a
        \tag{$\stThen$-\textsc{UnitExist\textsubscript{L}}} \quad \label{frameAx:then-unit-left} \\
        & \exists e \in E \colon a = a \stThen e
        \tag{$\stThen$-\textsc{UnitExist\textsubscript{R}}} \label{frameAx:then-unit-right} \\
        & (a \stThen b) \stThen c \eqIfDefined a \stThen (b \stThen c)
        \tag{$\stThen$-\textsc{Assoc}}\label{frameAx:then-ass}
        \end{align*}
    \end{minipage}
    \vspace{-.3cm}
    \begin{align*}
        & e \in E \,\&\, (a \stPlus e) \isDefined \,\,\onlyif (a \stPlus e) \stGeq a
        \tag{$\stPlus$-\textsc{UnitCoh}} \label{frameAx:plus-unit-coherent} \\
        & e \in E \,\&\, (a \stThen e) \isDefined \,\,\onlyif (a \stThen e) \stGeq a \quad
        \tag{$\stThen$-\textsc{UnitCoh\textsubscript{R}}}\label{frameAx:then-coherence-right} \\
        & e \in E \,\&\, e'\stGeq e \onlyif e' \in E \tag{\textsc{UnitClosure}} \label{frameAx:unit-closure} \\
        & (a \stPlus b) \isDefined \,\&~\, a \stGeq a' \,\&\, b \stGeq b' \onlyif (a' \stPlus b') \isDefined \,\&\, (a \stPlus b) \stGeq (a' \stPlus b')
        \tag{$\stPlus$-\textsc{DownClosed}}\label{frameAx:plus-down-closed} \\
        & (a \stThen b) \isDefined \,\&\, (a \stThen b) \stLeq c' \onlyif \exists a', b' \colon a' \stGeq a \,\&\, b' \stGeq b \,\&\, c' = (a' \stThen b')
        \tag{$\stThen$-\textsc{UpClosed}}\label{frameAx:then-up-closed} \\
        & (a_1 \stThen a_2) \stPlus (b_1 \stThen b_2) \eqIfDefined (a_1 \stPlus b_1) \stThen (a_2 \stPlus b_2)
    \tag{\textsc{RevExchange}} \label{frameAx:rev-exchange}
            \end{align*}
    \vspace*{-0.6cm}
    \caption{DIBI frame conditions (with implicit outermost universal quantifiers), where $\isDefined$ stands for `is defined', $\eqIfDefined$ means `equal when either side is defined'.}
    \label{fig:dibi-frame}
  \end{figure}
\end{definition}
The operations $\stThen$ and $\stPlus$ are referred to as the sequential and parallel compositions of states. Intuitively, $a \stSmaller b$ says that $a$ can be extended to $b$,  and $E$ is the set of states that act as units for these operations.

For capturing conditional independence, atomic propositions $\setAtomProp$ have the form $\kp{S}{T}$, for finite sets of variables $S, T$. Roughly, $\kp{S}{T}$ means the values of variables in $T$ only depend on that of $S$. We now recall the semantics of DIBI formulas, restricting to the fragment that is necessary for the current work.

%
\begin{definition}
\label{def:DIBI-validity-sem-kernel}
Given a DIBI model $\langle \frameA, \valV \rangle$, \emph{satisfaction $\satisfy{\valV}$} of \smallDIBI{}-formulas at $\frameA$-states is inductively defined as follows:
   \begin{center}
    \begin{tabular}{@{\!\!\!}l@{\quad}c@{\quad}l}
        $a \satisfy{\valV} I$ \  iff \ $a \in E$
        &&$a \satisfy{\valV} \top$ \ always\\
        $a \satisfy{\valV} (\kp{A}{B})$ & iff & $a \in \valV(\kp{A}{B})$ \\
        $a \satisfy{\valV} P \land Q$ & iff & $a \satisfy{\valV} P ~\text{and}~ a \satisfy{\valV} Q$ \\
        $a \satisfy{\valV} P \mand Q$ & iff & $\exists b_1, b_2$ such that $b_1 \stPlus b_2 \stLeq a$, $b_1 \satisfy{\valV} P$, $b_2 \satisfy{\valV} Q$ \\
        $a \satisfy{\valV} P \biThen Q$ & iff &
        $\exists b_1,
        b_2$ such that $b_1 \stThen b_2 = a$, $b_1 \satisfy{\valV} P$, $b_2 \satisfy{\valV} Q$
    \end{tabular}
   \end{center}
\end{definition}

For a concrete example of DIBI models, we review the probabilistic models.
Let $\setVal$ be a set of values, to which variables in $\setVar$ are assigned.
A \emph{memory over} a finite set of variables $X$ is a function $\memoryM \colon X \to \setVal$, and the \emph{memory space} over $X$ is the set of all memories over $X$,
denoted as $\memSpace{X}{\setVal}$, or $\memSpaceS{X}$ when $\setVal$ is clear.
Given a memory $\memoryM \in \memSpaceS{X}$ and a subset $U \subseteq X$, the memory $\memProj{\memoryM}{U} \colon U \to \setVal$ is the restriction of $\memoryM$ to the domain $U$.
Given a set $S$, $\dist S$ is the set of discrete distributions over $S$; that is, for any $\mu \in \dist S$, we have $\mu \colon S \to [0,1]$, the support $\support{\mu} = \{ s \in S \mid \mu(s) > 0 \}$ is finite, and $\sum_{s \in S} \mu(s) = 1$.
A dirac distribution $\mu$ on an outcome $s$ is the distribution $\mu$ such that
$\mu(s) = 1$, and $\mu(s') = 0$ for any $s' \neq s$.
Given a distribution $\mu$ in $\dist \memSpaceS{X}$, if $Y \subseteq X$,
we define the marginalisation of $\mu$ to $\dist \memSpaceS{Y}$, written as
$\proj_{Y} \mu$, by letting $(\proj_{Y} \mu) (\memoryM_Y) = \sum_{\memoryM \in \memSpaceS{X} \mid \memProj{\memoryM}{Y} = \memoryM_Y } \mu(\memoryM)$.

%
To capture conditional independence, we now introduce the notion of \emph{probabilistic input-preserving kernels}.
%
\begin{definition}[\cite{bao2021bunched}]
\label{def:probabilistic-kernel}
    A \emph{probabilistic input-preserving kernel} (or \emph{probabilistic kernel} for short) is a function $f \colon \memSpaceS{X}{} \to \dist \memSpaceS{Y}{}$ satisfying:
    \begin{enumerate*}[label=(\roman*)]
        \item $X \subseteq Y$, \label{item:prob-kernel-1}
        \item $\proj_{X} \circ f = \monadUnit{\dist}{\memSpaceS{X}}$,
        where $\monadUnit{\dist}{\memSpaceS{X}} (\memoryM)$ returns the dirac
        distribution over $\memoryM$.
        \label{item:prob-kernel-2}
    \end{enumerate*}%
    The set of all probabilistic kernels is denoted $\conProbKern$.
\end{definition}
In words, a probabilistic kernel $f$ maps a memory $\memoryM$ on $X$ to a distribution of memories on $Y \supseteq X$ whose support contains only memories $\memoryM'$ that faithfully extends $\memoryM$ (thus the name `input-preserving'). Alternatively, $f$ can be seen as a conditional distribution $\Pr(Y \!\mid\! X)$ where $Y \supseteq X$, such that $\Pr(Y = B \!\mid\! X = A)$ is nonzero only if $B$ restricted to the domain $X$ equals $A$.
%
%
\begin{definition}[Probabilistic model, \cite{bao2021bunched}]
\label{def:concrete-prob-frame}
The \emph{probabilistic frame $\probFrameBasedOn{\setVal}$ based on $\Values{}$} is a tuple $\lr{\conProbKern, \stLeq, \stPlus, \stThen,
\conProbKern}$
where $\stThen, \stPlus, \stLeq$ are defined for
arbitrary
    $f \colon \memSpaceS{X} \to \dist \memSpaceS{Y}$ and
    $g \colon \memSpaceS{Z} \to \dist \memSpaceS{W} $ as:
\vspace{-0.1cm}
\begin{itemize}
    \item
    the sequential composition $f \stThen g$ is defined iff $Y = Z$. In this case, $f \stThen g$ is of the form $\memSpaceS{X} \to \dist \memSpaceS{W}$, and given $\memoryM \in \memSpaceS{X}$, $(f \stThen g)(\memoryM)$ maps $\memoryN \in \memSpaceS{W}$ to $\sum_{\memoryL \in \support{f(\memoryM})} g(\memoryL)(\memoryN)$;
    \item
    the parallel composition $f \stPlus g$ is defined iff $X \cap Z =
    Y \cap W$. In this case, $f \stPlus g$ is of the form $\memSpaceS{X \cup Z} \to \dist \memSpaceS{Y \cup W}$ such that given $\memoryL \in \memSpaceS{X \cup Z}$ and $\memoryM \in \memSpaceS{Y \cup W}$,
    we have
    $(f \stPlus g)(\memoryL)(\memoryM) = f(\memProj{\memoryL}{X} )(\memProj{\memoryM}{Y}) \cdot g(\memProj{\memoryL}{Z})(\memProj{\memoryM}{W})$;
    \item the subkernel relation $f \stLeq g$ holds if there exist a finite set of variables $S$ and $h \in \conProbKern$ such that $g = \left( f \stPlus \monadUnit{\dist}{\memSpaceS{S}} \right) \stThen h$.
    %
\end{itemize}
\vspace{-0.1cm}
The \emph{probabilistic model based on $\setVal$} consists of the probabilistic frame $\probFrameBasedOn{\setVal}$ and the following \emph{natural valuation} $\natVal \colon \setAtomProp \to \powset(\conProbKern)$: given $(\kp{S}{T})$ and $f \colon \memSpaceS{X} \to \dist \memSpaceS{Y}$, $f \in \natVal(\kp{S}{T})$ iff there exists a probabilistic kernel $f' \colon \memSpaceS{X'} \to \dist \memSpaceS{Y'}$ such that $f' \sqsubseteq f$, $X' = S$ and $T \subseteq Y'$.
\end{definition}
We may simply write $\probFrame$ when the underlying set of values $\setVal$ is evident.

Next we give examples of probabilistic kernels and how they compose.
We will abbreviate a map from a variable $x$ to a value $c$ as $c_x$ and use the ket notation $\ket{\omega}$ for each probabilistic outcome $\omega$.

\begin{example}
\label{ex:prob-kernel}
    Consider variables $x, y, z$ that take values in $\Val = \Bool$.
    We define a map $f \colon \memSpaceS{ \{ z \} } \to \dist \memSpaceS{ \{ x, y, z \} }$ by:
    {\small
    \begin{align*}
        f( \oran{0_z}) & =
        \frac{1}{4} \ket{0_x, 0_y, \oran{0_z}}
        + \frac{1}{4} \ket{0_x, 1_y, \oran{0_z}}
          + \frac{1}{4} \ket{1_y, 0_y, \oran{0_z}}
        + \frac{1}{4} \ket{1_y, 1_y, \oran{0_z}}\\
        f( \purp{1_z}) & =
        \frac{1}{16} \ket{0_x, 0_y, \purp{1_z}}
        + \frac{3}{16} \ket{0_x, 1_y, \purp{1_z}}
         + \frac{3}{16} \ket{1_y, 0_y, \purp{1_z}}
        + \frac{9}{16} \ket{1_y, 1_y, \purp{1_z}}
    \end{align*}
    }%
    %
    The input memory (coloured) is preserved by $f$ so it is a probabilistic kernel.
    Then define $g_1 \colon \memSpaceS{ \{z\} } \to \dist\memSpaceS{ \{x, z\} }$ and $g_2 \colon \memSpaceS{ \{z\} } \to \dist\memSpaceS{ \{y, z\} }$  as:
    {\small
    \begin{align*}
        g_1 (\oran{0_z}) & = \frac{1}{2} \ket{0_x, \oran{0_z}} + \frac{1}{2} \ket{1_y, \oran{0_z}}
        & g_1 (\purp{1_z}) & = \frac{1}{4} \ket{0_x, \purp{1_z}} + \frac{3}{4} \ket{1_y, \purp{1_z}} \\
        g_2 (\oran{0_z}) & = \frac{1}{2} \ket{0_y, \oran{0_z}} + \frac{1}{2} \ket{1_y, \oran{0_z}}
        & g_2 (\purp{1_z}) & = \frac{1}{4} \ket{0_y, \purp{1_z}} + \frac{3}{4} \ket{1_y, \purp{1_z}}
    \end{align*}
    }%
    Both $g_1$ and $g_2$ are probabilistic kernels as well.
    The parallel composition $g_1 \stPlus g_2$ is defined since $\{ z \} \cap \{ z \} = \{x,z\} \cap \{y, z\}$; in fact, it is easy to verify that $g_1 \stPlus g_2 = f$.
    Moreover, $g_1$ and $g_2$ can be obtained by projecting the output of $f$ on $\{ x, z \}$ and $\{ y, z \}$, respectively, and we can show $g_1 \stLeq f$ and $g_2 \stLeq f$.
\end{example}
%
%
\section{DIBI models in Markov categories}
\label{sec:Markov-DIBI-model}
In this section we construct more abstract DIBI models based on categorical structures.
The starting point of our approach is a categorical characterisation of the concrete probabilistic models given above.
In the following, we begin by showing examples of how elements in that model can be reformulated in categorical terms and then formally present our categorical construction of DIBI models.

As we noted in~\Cref{sec:intro}, the probabilistic DIBI kernels can be identified as morphisms in the Kleisli category for the distribution monad $\Kl(\dist)$ (\Cref{def:dist-monad}); however, not all morphisms in $\Kl(\dist)$ are probabilistic DIBI kernels, so we need to define the extra conditions categorically.
Let $\catMemProb$ be the subcategory of $\Kl(\dist)$ where objects are restricted to memory spaces over $\setVal$. That is, the objects are memory spaces $\memoryM  \colon X \to \setVal$, and the morphisms are maps $f: \memSpaceS{X} \to \dist{\memSpaceS{Y}}$.
%
Then, probabilistic kernels are exactly those morphisms in the  $\catMemProb$ that satisfy the input-preserving condition in~\Cref{def:probabilistic-kernel}.
To visualise the probabilistic kernels, we depict them using string diagrams, which is possible because $\catMemProb$ is a subcategory of $\Kl(\dist)$ and $\Kl(\dist)$ has the monoidal structure(\Cref{appendix:sec-app}).
We also observe that the codomain of an input-preserving kernel $f \colon \memSpaceS{X} \klto \memSpaceS{Y}$ can be decomposed as $\memSpaceS{X} \times \memSpaceS{Y \setminus X}$,
Recall the probabilistic kernel $f$ from \Cref{ex:prob-kernel}. Since its codomain $\memSpaceS{\{x, y, z\}}$ can be decomposed as $\memSpaceS{\{x\}} \times \memSpaceS{\{y\}} \times \memSpaceS{\{z\}}$, we can draw it as follows:
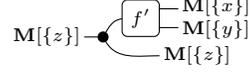
\[\input{kernel-eg-3.tikz}\]
Intuitively, $\input{copy-Mz.tikz}$ produces two copies of the value of $z$, and the values of $x$ and $y$ are computed from that of $z$ via $\input{f-prime.tikz}$, while the value of $z$ gets preserved through a straight wire in the bottom.
As in this example, such copy structure of $\Kl(\dist)$ enables us to capture the `input-preserving' condition of probabilistic kernels generally.


Next we want to express the sequential ($\stThen$) and parallel ($\stPlus$) compositions of probabilistic kernels categorically.
The former is exactly the sequential composition of morphisms in $\Kl(\dist)$.
The parallel composition, however, is \emph{not} the monoidal product $\tensor$ in $\Kl(\dist)$, which is the Cartesian product on objects. By definition, the monoidal product is total, while the parallel composition is partial. Even when the parallel composition is defined, the types of the resulting morphisms do not match. Suppose that the parallel composition of $f \colon \memSpaceS{X} \klto \memSpaceS{Y}$ and $g \colon \memSpaceS{U} \klto \memSpaceS{V}$ is defined, we have
\[
    f \stPlus g \colon \memSpaceS{ X \cup U } \klto \memSpaceS{ Y \cup V }
    \quad
    f \tensor g \colon \memSpaceS{X} \times \memSpaceS{U} \klto \memSpaceS{Y} \times \memSpaceS{V}
\]
The key difference is that parallel composition considers a single memory that can be projected into two pieces, while the monoidal product considers the cartesian product of two pieces of memory, no matter if they agree or not on overlapped variables.
To combine $\memSpaceS{X}$ and $\memSpaceS{U}$ into $\memSpaceS{X \cup U}$ categorically, we note that for disjoint $Z_1, Z_2$, $\memSpaceS{Z_1 \cup Z_2} \cong \memSpaceS{Z_1} \times \memSpaceS{Z_2}$, therefore
$\memSpaceS{X \cup U} \cong \memSpaceS{X\setminus U} \times \memSpaceS{X \cap U} \times \memSpaceS{U \setminus X}$.
Thus we can illustrate the parallel composition of two probabilistic kernels as in the following example.

%
%

\begin{example}
\label{eg:first-diagram-parallel}
    The probabilistic kernels $g_1$ and $g_2$ from \Cref{ex:prob-kernel} -- seen as $\Kl(\dist)$-morphisms -- are drawn as the first and second string diagram below respectively.
    \begin{equation*}
         \input{diagram-g1.tikz} \qquad \input{diagram-g2.tikz} \qquad  \scalebox{.9}{\input{g1-g2-para-comp.tikz}}
    \end{equation*}
    where $g_1' \colon \memSpaceS{\{z\}} \klto \memSpaceS{\{x\}}$ and $g_2' \colon \memSpaceS{\{z\}} \klto \memSpaceS{\{y\}}$ represent the conditional distributions obtained by suitable projections of $g_1$ and $g_2$, respectively. The parallel composition $g_1 \stPlus g_2$ is given by the rightmost string diagram above. 
\end{example}
We omit the formal string diagrammatic definitions here as they can be easily derived from their counterparts in the generic construction of DIBI models, defined later in \Cref{def:kernel-composition}.

Towards that categorical construction of DIBI models,
we also want to generalise the concept of memory spaces
$\memSpaceS{X}$, which were customised for reasoning about probabilistic programs and relational databases.
We observe that the side conditions of the parallel and sequential compositions are all based on comparing the set of variables in the (co)domains, so they only depend on the $X$ part in $\memSpaceS{X}$.
This motivates us to define DIBI states as morphisms in a category whose objects are made of variables in~\Cref{def:category-with-assignment}.
To express finite sets of variables and the union of disjoint such sets in a monoidal category,
in the following we will represent finite sets of variables as lists. We impose a linear order $\varLeq$ on $\setVar$ such that indexed variables  inherit the order of their indices, e.g., $x_1 \varLeq x_2 \varLeq x_3$.
Let $x \varLessThan y$ abbreviate for $x \varLeq y$ and $x \neq y$.
Then, finite sets of variables can be represented as finite lists of variables ordered by $\varLessThan$, via a translation that we write as $\setToList{\cdot}$. For instance, $\setToList{\{ x_3, x_1, x_3, x_4 \}} = [x_1, x_3, x_4]$. 

Throughout the rest of the section we fix a Markov category $\lrangle{\catC, \tensor, \tenUnit}$, and assignment $\obChoice \colon \Var \to \ob{\catC}$ of a $\catC$-object to each $x \in \Var$.
\begin{definition}
\label{def:category-with-assignment}
    Let $\catWithVarName{\catC}{\obChoice}$ be the symmetric monoidal category whose objects are finite lists of variables, and morphisms $[x_1, \dots, x_m] \to [y_1, \dots, y_n]$ are $\catC$-morphisms $\obChoice(x_1) \tensor \cdots \tensor \obChoice(x_m) \to \obChoice(y_1) \tensor \cdots \tensor \obChoice(y_n)$.
    Sequential composition is defined as in $\catC$. The identity on $[x_1, \dots, x_m]$ is $\idMor_{\obChoice x_1 \tensor \cdots \tensor \obChoice x_m}$.
    The monoidal product in $\catWithVarName{\catC}{\obChoice}$ -- which we also write as $\tensor$ with abuse of notation -- is list concatenation on objects, and monoidal product in $\catC$ on morphisms.
\end{definition}
%
%
Since $\catC$ is a Markov category, it follows immediately that $\catWithVarName{\catC}{\obChoice}$ is also a Markov category.
Sometimes we restrict ourselves to a uniform assignment $\obChoice$; that is, for some fixed $\obC \in \ob{\catC}$, $\obChoice(x) = \obC$ for all $x \in \setVar$. This is in line with the scenario where a fixed value space $\setVal$ is used for all variables (see \Cref{def:probabilistic-kernel}).
In this case, we write $\catWithVarName{\catC}{\obChoice}$ as $\catWithVarName{\catC}{\obC}$ to emphasise the uniform value of the assignment. This category can be seen as the full subcategory of $\catC$ freely generated by $\obC$, but with each occurrence of the generating object named by a variable. The next example shows how the construction in \Cref{def:category-with-assignment} selects morphisms of $\Kl(\dist)$ that act on memory spaces, among which we have all the probabilistic kernels.

%
\begin{example}
\label{ex:cat-with-var-name}
   Let $\catC$ be $\Kl(\dist)$, and $\obChoice \colon \setVar \to \ob{\Kl(\dist)}$ be the constant function $x \mapsto \setVal$ for all $x \in \setVar$. Then there is a full and faithful embedding functor $\iota \colon \catMemProb \to \catWithVarName{\Kl(\dist)}{\obChoice}$: on objects, given a set $X$, $\iota(\memSpaceS{X}) = \setToList{X}$; on morphisms, given $f \colon \memSpaceS{X} \to \dist \memSpaceS{Y}$ with $X = \{ x_1, \dots, x_m \}$ and $Y = \{ y_1, \dots, y_n \}$, its image $\iota(f) \colon X \to Y$ is the composed map $\setVal^{m} \xrightarrow{\cong} \memSpaceS{X} \xrightarrow{f} \dist \memSpaceS{Y} \xrightarrow{\dist \cong} \dist \setVal^n$, where the isomorphisms are, e.g., $\memSpaceS{Y} \xrightarrow{\cong} \memSpaceS{y_1} \times \cdots \times \memSpaceS{y_n} \xrightarrow{\cong^{\obChoice} } \setVal^n$, using the valuation $\obChoice(y_j) = \setVal$.
\end{example}
Now that we have abstracted the concept of memory spaces used in concrete DIBI models as objects in Markov categories, the states of DIBI models --- the role that in the probabilistic models is filled by probabilistic kernels --- takes the form of morphisms in these categories. Next we  identify  those $\catC$-morphisms that constitute the states of DIBI models. We call them the \emph{input-preserving kernels} in $\catC$, written $\setSemKernel{\catWithVarName{\catC}{\obChoice}}$.


\begin{definition}
\label{def:cat-IP-kernel}
    A $\catWithVarName{\catC}{\obChoice}$-morphism $f \colon [x_1, \dots, x_m] \to [y_1, \dots, y_n]$ is a \emph{$\catWithVarName{\catC}{\obChoice}$ input-preserving kernel} (or \emph{$\catWithVarName{\catC}{\obChoice}$-kernel} for short) if $x_1 \varLessThan \cdots \varLessThan x_m$, $y_1 \varLessThan \cdots \varLessThan y_n$, and $f$ can be decomposed as follows, where $\sigma$ is rewiring:
    \begin{equation}
    \label{eq:C-kernel}
        \input{diagram-f.tikz} = \scalebox{0.9}{\input{C-kernel-1.tikz}}
    \end{equation}
\end{definition}
In words, a $\catWithVarName{\catC}{\obChoice}$-kernel is a morphism whose interfaces are essentially finite sets of variables, such that the input is preserved as part of the output (through the upper leg of those $\copier$s). The map $f'$ in \eqref{eq:C-kernel} is referred to as the nontrivial part of the input-preserving kernel.
It follows from \Cref{def:cat-IP-kernel} that, for a $\catWithVarName{\catC}{\obChoice}$-kernel, its codomain $[ y_1, \dots, y_n]$ always subsumes its domain $[ x_1, \dots, x_m ]$; also, $u_1, \dots, u_k $ are precisely those $y_j$s that are not among these $x_i$s. Since the (co)domains of $\catWithVarName{\catC}{\obChoice}$-kernel are list presentation of sets, we also write the types of $\catWithVarName{\catC}{\obChoice}$-kernels using the corresponding sets, e.g., in \eqref{eq:C-kernel}, $f \colon \{ x_1, \dots, x_m \} \to \{ y_1, \dots, y_n \}$.

Next we define compositions on input-preserving kernels, generalising what we have seen in Example~\ref{eg:first-diagram-parallel} for the probabilistic models.

\begin{definition}[Compositions]
\label{def:kernel-composition}
Given arbitrary $\catWithVarName{\catC}{\obChoice}$-kernels $f \colon X \to Y$ and $g \colon U \to V$ as in \Cref{fig:diagram-f-and-g}, their \emph{sequential composition} $f \stThen g$ is defined iff $\codom{f} = \domain{g}$, in which case $f \stThen g = g \after f$. Their \emph{parallel composition} $f \stPlus g$ is defined iff $X \cap U = Y \cap V$. Assume $\listL = \setToList{ X \cap U }$, $\listL_1 = \setToList{ X \setminus (X \cap U) }$, $\listL_2 = \setToList{ U \setminus (X \cap U) }$, $\listK_1 = \setToList{Y \setminus (Y \cap V)}$, and $\listK_2 = \setToList{V \setminus (Y \cap V)}$, then $f \stPlus g \colon X \cup U \to Y \cup V$ is defined as in~\Cref{fig:diagram-c-plus-d}, where all the $\sigma_i$s are rewiring morphisms for making the input and output variables $\varLessThan$-ordered.
%

\begin{figure}[hbt]
    \centering \vspace*{-.8cm}
    \begin{subfigure}[b]{0.6\textwidth}
        \[
            \scalebox{.88}{\input{def-kernel-para-comp-3.tikz}} ~
            \scalebox{.88}{\input{def-kernel-para-comp-4.tikz}}
        \]
           \vspace*{-.4cm}
        \caption{$\catWithVarName{\catC}{\obChoice}$-kernels $f$ and $g$.}
        \label{fig:diagram-f-and-g}
    \end{subfigure}
    \begin{subfigure}[b]{0.35\textwidth}
        \[
            \scalebox{.88}{\input{def-kernel-para-comp-5.tikz}}
        \]
           \vspace*{-.4cm}
        \caption{$f \stPlus g$}
        \label{fig:diagram-c-plus-d}
    \end{subfigure}
      \vspace*{-.2cm}
    \caption{Parallel composition of $\catWithVarName{\catC}{\obChoice}$.}
    \label{fig:para-compose-diagram}
    \vspace*{-.6cm}
\end{figure}
\end{definition}
Note here a benefit of the diagrammatic representation: we can easily identify the memory overlap $\memSpaceS{X \cap Y}$, as it is depicted a separate wire; with traditional syntax, we would need to apply associativity and commutativity to extract it from $\memSpaceS{X \cup Y}$.
It is easy to see that kernels are closed under compositions. Also, for curious readers, we note that $\catWithVarName{\catC}{\obChoice}$-kernels with their parallel compositions form a \emph{partially monoidal category}~\cite{balco2018partially}. Next we define the subkernel relation.


\vspace{0.2cm}
\noindent\begin{minipage}{0.55\textwidth}
\begin{definition}[Subkernel]
\label{def:kernel-leq}
Given two $\catWithVarName{\catC}{\obChoice}$-kernels $f$ and $g$, we say $f$ is a
\emph{subkernel} of $g$ -- denoted as $f \stLeq g$ -- if there exist $z_1, \cdots, z_n \in \setVar$, a $\catWithVarName{\catC}{\obChoice}$-kernel $h$, and rewiring morphisms $\sigma_1, \sigma_2$ such that $g$ can expressed as on the right-hand side.
\end{definition}
\end{minipage}%
\begin{minipage}{0.45\textwidth}
    \begin{equation*}\label{eq:def-kernel-leq}
    g = \input{def-subkernel.tikz}
    \end{equation*}
\end{minipage}
\vspace{0.2cm}

The subkernel relation is transitive and reflexive, which can be shown simply by manipulations of the string diagram.
We are finally able to state the main result of this section: $\catWithVarName{\catC}{\obChoice}$-kernels and their compositions form a DIBI frame.
\begin{restatable}{theorem}{ThmInstDibi}
    \label{thm:inst-DIBI}
    $\frameSemKernel{\catWithVarName{\catC}{\obChoice}} = \lrangle{ \setSemKernel{\catWithVarName{\catC}{\obChoice}}, \stLeq, \stPlus, \stThen, \setSemKernel{\catWithVarName{\catC}{\obChoice}} }$ is a DIBI frame.
\end{restatable}

%
\noindent Also, under the natural valuation $\natVal$, a $\catWithVarName{\catC}{\obChoice}$-kernel $f \colon X \to Y$ satisfies $\kp{S}{T}$ iff there is a subkernel $(f' \colon X' \to Y') \stLeq f$ such that $X' = S$ and $Y' \supseteq T$. Thus:
\begin{corollary} $(\frameSemKernel{\catWithVarName{\catC}{\obChoice}}, \natVal)$ is a DIBI model.
\end{corollary}


We will see in \Cref{sec:applications} how to use this categorical construction to derive a wide range of DIBI models. Moreover, it also helps to extract properties of the underlying category that are essential to a specific feature of a DIBI model. Here is an example.

\begin{restatable}{proposition}{PropSubkernelUnique}
\label{prop:subkernel-unique}
    If $\catC$ further satisfies that for arbitrary morphisms $f, g$ and object $\obD$, $f \tensor \discardMor_{\obD} = g \tensor \discardMor_{\obD}$ implies $f = g$, then subkernel is unique given its type in the following sense: if $\catWithVarName{\catC}{\obChoice}$-kernels $f_1, f_2 \colon U \to V$ are both subkernels of $g$, then $f_1 = f_2$.
\end{restatable}

Note that the uniqueness of subkernels has been observed already in the context of probabilistic and relational models, see~\cite[Sect.~IV]{bao2021bunched}. \Cref{prop:subkernel-unique} reveals the general conditions under which we have this uniqueness for a wider class of DIBI models. In particular, we will see in \Cref{subsec:prob-DIBI-frame} how the probabilistic models meets the conditions of \Cref{prop:subkernel-unique}.


%


\section{Examples}
\label{sec:applications}
In this section we provide concrete instances of the categorical construction in \Cref{sec:Markov-DIBI-model}. The first example concerns the probabilistic DIBI models.
The remaining examples are new DIBI models. Some of them have been suggested in the DIBI paper~\cite{bao2021bunched}, yet not materialised due to the complexity involved in stating each component and verifying the frame conditions. Within our framework, these steps become much easier to perform.
\subsection{Probabilistic (and Relational) DIBI Models}
\label{subsec:prob-DIBI-frame}
As we sketched in \Cref{eg:first-diagram-parallel} and \Cref{ex:cat-with-var-name},
the probabilistic DIBI kernels and
$\lrangle{\frameSemKernel{\Kl(\dist)}, \natVal}$ input-preserving kernels
correspond to each other.
We now formally show that the probabilistic DIBI model in  \Cref{def:concrete-prob-frame} can be recovered from the categorical DIBI model $\lrangle{\frameSemKernel{\Kl(\dist)}, \natVal}$. Since both models are equipped with the natural valuation $\natVal$, we focus on the frame part.
To make the correspondence precise, we introduce the category of DIBI frames, as hinted in \cite[Sect.~III]{bao2021bunched}.
\begin{definition}
\label{def:cat-DIBI-frame}
    In \emph{the category of DIBI frames $\catDibiFrame$}, objects are DIBI frames; morphisms $f \colon \lr{S, \stLeq_{S}, \stPlus_{S}, \stThen_{S}, \unitSet_{S}} \to \lr{T, \stLeq_{T}, \stPlus_{T}, \stThen_{T}, \unitSet_{T}}$ are functions $f \colon S \to T$ that respect all the relations and partial operations: for arbitrary $s, s' \in S$,
    \begin{itemize}
        \item $s \stLeq_{S} s'$ implies $f(s) \stLeq_{T} f(s')$;
        \item if $s \operateVar_{S} s'$ is defined, then $f(s) \operateVar_{S} f(s')$ is defined, and $f(s) \operateVar_{T} f(s') = f(s \operateVar_{S} s')$, for $\operateVar \in \{ \stPlus, \stThen \}$;
        \item $s \in \unitSet_{S}$ implies $f(s) \in \unitSet_{T}$.
    \end{itemize}
\end{definition}
It turns out that the function $\iota$ introduced in~\Cref{ex:cat-with-var-name} extends to an isomorphism of DIBI frames from $\probFrameBasedOn{\setVal}$ to $\frameSemKernel{\catWithVarName{\Kl(\dist)}{\setVal}}$.
\begin{restatable}{proposition}{PropProbModelEq}
\label{prop:prob-KlD-model-eq}
$\probFrameBasedOn{\setVal} \cong \frameSemKernel{\catWithVarName{\Kl(\dist)}{\setVal}}$.
\end{restatable}
\begin{example}\label{ex:relate-prob-and-dist-kernels}
    The probabilistic kernel $g_1 \colon \memSpaceS{\{z\}} \to \dist \memSpaceS{ \{ x, z\} }$ from \Cref{ex:prob-kernel} corresponds to the following $\catWithVarName{\Kl(\dist)}{\{0, 1\}}$-kernel $h_1 \colon [z] \to [x, z]$ -- i.e., a $\Kl(\dist)$-morphism $\{0, 1\} \klto \{0, 1\}^2$ -- where: $\oran{0} \mapsto \frac{1}{2} \ket{0, \oran{0}} + \frac{1}{2} \ket{1, \oran{0}}$, $\purp{1} \mapsto \frac{1}{4} \ket{0, \purp{1}} + \frac{3}{4} \ket{1, \purp{1}}$.
    Diagrammatically, $h_1$ is of the form $\input{eg-prob-kern-1.tikz}$,
    where $h'_1 \colon [z] \to [x]$ is the map such that $\oran{0} \mapsto \frac{1}{2} \ket{0} + \frac{1}{2} \ket{1}$ and $\purp{1} \mapsto \frac{1}{4} \ket{0} + \frac{3}{4} \ket{1}$
\end{example}

Similarly, the relational DIBI model from \cite{bao2021bunched} with the value space $\setVal$ can be shown to be isomorphic to $\frameSemKernel{\catWithVarName{\Kl(\ipowset)}{\setVal}}$, where $\ipowset$ is the nonempty powerset monad.

\subsection{Stochastic DIBI Models}
\label{subsec:stoch-DIBI-frame}
%

Using our categorical construction, we can derive a notion of DIBI model for continuous probabilistic (stochastic) processes, not previously considered. This is of interest because, as we will later show in~\cref{sec:CI}, it allows to capture conditional independence for continuous probability using DIBI formulas.
We take as underlying category $\catStoch$ of stochastic processes, defined as the Kleisli category $\Kl(\monadGiry)$ for the Giry monad on measurable spaces -- see \Cref{appendix:sec-app} for a full definition. Since $\monadGiry$ is an affine symmetric monoidal monad, $\catStoch$ is a Markov category~\cite{fritz2020markov}.
Applying~\Cref{thm:inst-DIBI} to $\catC = \catStoch$,  we get DIBI frames based on stochastic processes.
\begin{proposition}
\label{prop:stoch-DIBImodel}
Given an arbitrary map $\obChoice: \Var \to \ob{\catMeas}$,
$\frameSemKernel{\catWithVarName{\catStoch)}{\obChoice}} = \lrangle{ \setSemKernel{\catWithVarName{\catStoch}{\obChoice}}, \stLeq, \stPlus, \stThen, \setSemKernel{\catWithVarName{\catStoch}{\obChoice}} }$
is a DIBI frame.
\end{proposition}
We call $\frameSemKernel{\catWithVarName{\catStoch}{\obChoice}}$ the \emph{stochastic DIBI frame} based on $\theta$ and elements in $\setSemKernel{\catWithVarName{\catStoch}{\obChoice}}$ stochastic kernels.

\begin{example}\label{ex:stoch-kernels}
    We show a representation of the \emph{Box-Muller transformation} using stochastic kernels.
    Consider $\obChoice$ that maps all variable names to the
    Borel $\sigma$-algebra over reals $(\real, \borel(\real))$.
    Define stochastic kernels $g_1 \colon \empset \to \{u\}$ and $g_2 \colon \empset \to \{w\}$ -- both standing for $\catStoch$-morphisms $(\setOne, \{ \empset, \setOne \}) \to (\real, \borel(\real))$, or equivalently, a probabilistic measure on $(\real, \borel(\real))$ -- by $g_i(\bullet) = \uniform(0, 1)$ for $i= 1, 2$, where $\uniform(0, 1)$ is the uniform measure over the interval $(0, 1)$. Such a uniform measure over infinite outcomes is
    not possible in the discrete probabilistic DIBI model.
    Define another stochastic kernel $f: \{u, w\} \to \{u, w, x, y\}$ where the value of $x, y$
    are determined by the value of $u, w$:
    \[
        f(u \mapsto v_u, w \mapsto v_w)
        = \dirac{v_u, v_w, \left( \sqrt{-2 \ln u} \cdot \cos (2 \pi w)\right)_x,  \left( \sqrt{-2 \ln u} \cdot \sin (2 \pi w)\right)_y}.
    \]
    Then $h = (g_1 \stPlus g_2) \stThen f$ gives a stochastic kernel $\empset \to \{u, w, x, y\}$.
    Box-Muller transformation says that $x$ and
    $y$ are independent in $h(\trivList)$, despite their seemingly dependence on $u$ and $w$. We will explain more on this in \Cref{ex:box-muller-indep}.
\end{example}

\noindent\textbf{Comparison with \textsf{Lilac}~\cite{li2023lilac}.}
Our stochastic DIBI models can be used to reason about independence and conditional probabilities in continuous distributions. A recent work \textsf{Lilac} by Li et al.~\cite{li2023lilac} proposed a BI model for the same goal, yet with some crucial differences in the set-up.

First, the states in Lilac's BI model are probabilistic space fragments of a fixed sample space, and their variables are mathematical random variables that deterministically map elements in the sample space to values.
In comparison, we treat variables as names that can be associated to values or distributions.
Our stochastic kernels -- though not using an ambient sample space -- can encode their set-up: we can devise a special variable $\Omega$ for `the sample space,' and deterministic kernels from $\Omega$ to other variables encode random variables.

Second, to reason about conditional probabilities, \textsf{Lilac} want probability spaces to be disintegrable with respect to well-behaved random variables.
To achieve that, they require probability spaces in their model to be extensible to Borel spaces, since disintegration works nicer in Borel spaces.
By working with kernels, which already represents conditional probability spaces, we do not need to impose disintegratability on our DIBI states to reason about conditional probabilities.
For instance, while disintegration is not always possible in the category $\catStoch$, we
can still construct a DIBI model based on $\catStoch$.

\medskip
\noindent\textbf{Other measure-theoretic probabilistic DIBI models.}
The category $\catStoch$ is not the only Markov category for
measure-theoretic probability. Another choice is $\catBorelStoch$, a subcategory of $\catStoch$ obtained by
restricting to standard Borel spaces as objects. It has some nice properties that $\catStoch$ does not satisfy, such as having conditionals as mentioned above.
$\catBorelStoch$ is also a Markov category and we can easily instantiate a DIBI model.

\begin{proposition}
\label{prop:borelstoch-DIBImodel}
Given any map $\obChoice \colon \setVar \to \ob{\catBorelStoch}$,
$\frameSemKernel{\catWithVarName{\catBorelStoch}{\obChoice}}$ defined as
$\lrangle{ \setSemKernel{\catWithVarName{\catBorelStoch}{\obChoice}}, \stLeq, \stPlus, \stThen, \setSemKernel{\catWithVarName{\catBorelStoch}{\obChoice}} }$
is a DIBI frame.
\end{proposition}

The study of measure theory is also intertwined with topology,
and another category for measure-theoretic probability is
the Kleisli category of the \emph{Radon monad} $\monadRadon$
based on the category of compact Hausdorff spaces $\catHaus$ and continuous maps (cf. \Cref{appendix:sec-app}), which we will denote as $\Kl_{\catHaus}(\monadRadon)$. $\Kl_{\catHaus}(\monadRadon)$ is also a Markov category~\cite{fritz2020markov}, so \Cref{thm:inst-DIBI} applies.
\begin{proposition}
\label{prop:Radon-DIBImodel}
Given any map $\obChoice: \Var \to \ob{\Kl_{\catHaus}(\monadRadon)}$,
$\frameSemKernel{\catWithVarName{\Kl_{\catHaus}(\monadRadon)}{\obChoice}}$ defined as $ \lrangle{ \setSemKernel{\catWithVarName{\Kl_{\catHaus}(\monadRadon)}{\obChoice}}, \stLeq, \stPlus, \stThen, \setSemKernel{\catWithVarName{\Kl_{\catHaus}(\monadRadon)}{\obChoice}} }$
is a DIBI frame.
\end{proposition}

A measure-theoretic Markov category not formed as Kleisli categories is the Gaussian probability category
$\catGauss$~\cite{fritz2020markov}.  Its objects are natural numbers, and a morphism $n \to m$ is a tuple $(M, \cov, \mu)$ representing the function $f: \real^{n} \to \real^{m}$ with $f(v) = M \cdot v + \xi$, where $\xi$ is the Gaussian noise with mean $\mu$ and covariance matrix $\cov$. Its monoidal product is addition $+$ on the objects and vector concatenation on morphisms.
$\catGauss$ differs from $\catStoch, \catBorelStoch$ and $\Kl_{\catHaus}(\monadRadon)$ in that it does not arise as the Kleisli category associated to some monad. But since it is a Markov category, we can again instantiate DIBI models based on $\catGauss$.
\begin{proposition}
\label{prop:Gauss-DIBImodel}
Given any map $\obChoice: \Var \to \ob{\catGauss}$,
$\frameSemKernel{\catWithVarName{ \catGauss}{\obChoice}}$ defined as $\lrangle{ \setSemKernel{\catWithVarName{ \catGauss}{\obChoice}}, \stLeq, \stPlus, \stThen, \setSemKernel{\catWithVarName{\catGauss}{\obChoice}} }$
is a DIBI frame.
\end{proposition}


\subsection{Syntactic DIBI Models}
\label{subsec:syn-DIBI-frame}
The DIBI models defined so far all have kernels defined by some processes over memory spaces. It is worth considering a different flavour: purely formal, syntactically generated DIBI models.
We start by defining the underlying category.
\begin{definition}
 $\catSynVar$ is the Markov category freely generated as follows:
    \begin{itemize}
        \item the generating objects are variables in $\setVar$;
        \item there is exactly one generating morphism of type $[u_1, \dots, u_m] \to [v_1, \dots, v_n]$ for distinct variables $u_1 \varLessThan \cdots \varLessThan u_m$ and $v_1 \varLessThan \cdots \varLessThan v_n$, written as string diagrams of the form $\input{gen-mor.tikz}$.
    \end{itemize}
\end{definition}
In words, $\catSynVar$-objects are finite lists of variables (without the requirements of duplicate-free or $\varLeq$-ordered); morphisms are diagrams freely concatenated using $\begin{tikzpicture}
	\begin{pgfonlayer}{nodelayer}
		\node [style=none] (0) at (-0.5, 0) {};
		\node [style=none] (1) at (0.5, 0) {};
	\end{pgfonlayer}
	\begin{pgfonlayer}{edgelayer}
		\draw (0.center) to (1.center);
	\end{pgfonlayer}
\end{tikzpicture}
$, $\copier$, $\discarder$, $\swap$, and $\input{gen-mor.tikz}$, taken equivalence class modulo the Markov category equations.
The syntactic DIBI frame is based on the category $\catWithVarName{\catSynVar}{\id}$, where $\id \colon \setVar \to \ob{\catSynVar}$ is the identity function.

\begin{proposition}

    $\frameSynKernel = \lrangle{ \setSemKernel{\catWithVarName{\catSynVar}{\id}}, \stLeq, \stPlus, \stThen, \setSemKernel{\catWithVarName{\catSynVar}{\id}} }$ is a DIBI frame.
\end{proposition}
Equipped with the natural valuation $\natVal$, one obtains a DIBI model $\lrangle{\frameSynKernel, \natVal}$.
We postpone an example of $\catWithVarName{\catSynVar}{\id}$-kernels till \Cref{sec:CI}, \Cref{ex:different-CI}, in which $\catWithVarName{\catSynVar}{\id}$-kernels are used to distinguish two notions of conditional independence in Markov categories.

An interesting question for future work is how to extend the syntactic DIBI model to a term model. Typically being initial objects in categories of models, term models play an important role in proving completeness and defining categorical semantics for formal systems, including algebraic theories~\cite{lawvere1963functorial}, logics~\cite{tarski1983logic}
(e.g., Lindenbaum–Tarski algebras)
and type theories~\cite{jacobs1994semantics,hofmann1995interpretation}.
A term model for DIBI could lead to a sound and complete axiomatisation of the specific version of DIBI logic in this paper, whose atomic propositions take the form of $\kp{S}{T}$.

%
%
\section{Conditional independence}
\label{sec:CI}
DIBI logic is designed for reasoning about CI. The prior work~\cite{bao2021bunched} shows that, conditional independence in the discrete probabilistic models and join dependency in the relational  models can be characterised by the same class of DIBI formulas.
Generalising this result, in this section we define a notion of CI on $\catWithVarName{\catC}{\obChoice}$-kernels based on formula satisfaction.
Since $\catWithVarName{\catC}{\obChoice}$ is a Markov category, we can compare our logical notion of CI with existing categorical definitions of CI in Markov categories~\cite{cho2019disintegration,fritz2020markov}.

Fix a Markov category $\catC$ and a map $\obChoice \colon \setVar \to \ob{\catC}$. We define CI in the DIBI model $\lrangle{\frameSemKernel{\catWithVarName{\catC}{\obChoice}}, \natVal}$. 

\begin{definition}[Conditional Independence]
\label{def:DIBI-CI}
For any mutually disjoint finite sets of variables $W, X, Y, U$, $X$ and $Y$ are \emph{DIBI conditionally independent given $W$} in a $\catWithVarName{\catC}{\obChoice}$-kernel\footnote{Note that $\catWithVarName{\catC}{\obChoice}$-kernels with domain $\empset$ are not to be thought of as maps with empty domains. For instance, $\catWithVarName{\Kl(\dist)}{\obChoice}$-kernels of the form $\empset \to \{x, y\}$ corresponds to $\Kl(\dist)$-morphisms $\setOne \klto \obChoice(x) \times \obChoice(y)$, which denote distributions over $x, y$.}$f \colon \empset \to W \cup X \cup Y \cup U$ if
\begin{equation}\label{eq:CI}
f \satisfy{\natVal} (\kp{\empset}{W}) \biThen ( (\kp{W}{X}) \mand (\kp{W}{Y}) ).
\end{equation}
In this case, we write $\dibiCI{X}{Y}{W}$.
\end{definition}


Let us unfold what~\eqref{eq:CI} means. Under the natural valuation $\natVal$, an atomic proposition of shape $\kp{S}{T}$ encodes the dependence of $T$ on $S$: formally, a $\catWithVarName{\catC}{\obChoice}$-kernel $f \colon X \to Y$ satisfies $\kp{S}{T}$ iff $f$ contains some subkernel $f' \colon S \to Y'$ such that $T \subseteq Y'$.
So the formula \eqref{eq:CI} requires that $f$ is a kernel with empty domain that can be decomposed as $f \sqsupseteq f_0 \stThen (f_1 \stPlus f_2)$,
where $f_0$ determines the value on $W$,
$f_1$ determines the value on $X$ given the value on $W$, $f_2$ determines the value on $Y$ given the value on $W$, and $f_1$ and $f_2$ do so independently of each other.


We illustrate the formula with examples in the discrete probabilistic DIBI model and the stochastic DIBI model.
\begin{example}
\label{ex:DIBI-CI}
    In the setting of \Cref{ex:prob-kernel}, consider the probabilistic kernel $h \colon \memSpaceS{\empset} \to \dist \memSpaceS{\{x, y, z\}}$ representing the following distribution:
    \begin{linenomath*}
    {\small\begin{align*}
        \mu = & \frac{1}{8} \ket{0_x, 0_y, \oran{0_z}}
        + \frac{1}{8} \ket{0_x, 1_y, \oran{0_z}}
        + \frac{1}{8} \ket{1_y, 0_y, \oran{0_z}}
        + \frac{1}{8} \ket{1_y, 1_y, \oran{0_z}} \\
        & + \frac{1}{32} \ket{0_x, 0_y, \purp{1_z}}
        + \frac{3}{32} \ket{0_x, 1_y, \purp{1_z}}
        + \frac{3}{32} \ket{1_y, 0_y, \purp{1_z}}
        + \frac{9}{32} \ket{1_y, 1_y, \purp{1_z}}
    \end{align*}}%
    \end{linenomath*}
    Then $h \satisfy{\natVal} (\kp{\empset}{\{z\}}) \biThen ( (\kp{\{z\}}{\{z, x\}}) \mand (\kp{\{z\}}{\{z, y\}}) )$, because $h = h_0 \stThen f = h_0 \stThen (f_1 \stPlus f_2 )$, where $h_0$ denotes the uniform distribution $\frac{1}{2} \ket{0_z} + \frac{1}{2} \ket{1_z}$.
\end{example}

\begin{example}
\label{ex:box-muller-indep}
    Define $g_1, g_2, f, h$ as in \Cref{ex:stoch-kernels}. We want to assert that variables $x$ and $y$ are independent in the distribution
    constructed by Box-Muller Transform.
    Independence is a special case of conditional independence in which the set of conditioned variables is empty.
    Thus, the goal is to assert $(\kp{\empset}{\empset}) \biThen ( (\kp{\empset}{\{x\}}) \mand (\kp{\empset}{\{y\}}) )$ -- equivalently, $(\kp{\empset}{\{x\}}) \mand (\kp{\empset}{\{y\}})$.

    Define $h_1: \empset \to \{x\}$ and $h_2: \empset \to \{y\}$
    both as the standard normal distribution $\normal(0,1)$.
    Clearly $h_1 \satisfy{\natVal} \kp{\empset}{\{x\}}$ and $h_2 \satisfy{\natVal} \kp{\empset}{\{y\}}$.
    Moreover, some non-trivial calculations would show that
    $ (h_1 \stPlus h_2)\stLeq h$, and consequently we have $h \satisfy{\natVal} (\kp{\empset}{\{x\}}) \mand (\kp{\empset}{\{y\}})$ by definition.
\end{example}

Since the categorical DIBI models are based on Markov categories, we compare our logical notion of CI on kernels with the canonical notion of CI in Markov categories, which defines CI as decomposability of morphisms. In \Cref{def:display-CI-plain,def:proc-markov-CI,def:extended-superset-CI}, a Markov category $\catX$ is fixed.


\vspace{0.2cm}
\noindent%
\begin{minipage}{0.58\textwidth}
\begin{definition}
\label{def:display-CI-plain}
An $\catX$-morphism $s \colon
\tenUnit \to \obW \tensor \obX \tensor \obY$ \emph{displays the
conditional independence} of $\obX$ and $\obY$ given $\obW$ if there exist $\catX$-morphisms $s_{\obW} \colon \tenUnit \to \obW$, $g_{\obX} \colon \obW \to \obX$, $g_{\obY} \colon \obW \to
\obY$ such that equation on the right holds. We write this as $\displayCI{\obX}{\obY}{\obW}$.
\end{definition}
\end{minipage}%
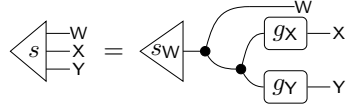
\begin{minipage}{0.38\textwidth}
  \[
      \ \ \ \scalebox{0.95}{\input{display-CI-plain-1.tikz}} = \scalebox{0.95}{\input{display-CI-plain-2.tikz}}
    \]
    \vspace{\baselineskip}
\end{minipage}
\vspace{0.2cm}
%


In the context of DIBI models, \Cref{def:display-CI-plain} restricts to stating the conditional independence of $X$ and $Y$ given $W$ in $\catWithVarName{\catC}{\obChoice}$-kernels of the form $\empset \to W \cup X \cup Y$. In particular, no extra variable in the kernel's codomain is allowed.

\begin{example}
We show an example of this notion of CI in the Markov category $\catGauss$.
Consider morphism $s: \empset \to \{w, x, y\}$ specified by the tuple \begin{align*}
    \left(\initialMap,
    \cov = \begin{bsmallmatrix}
           1 &1 &1\\
           1 &2 &1\\
           1 &1 &2
    \end{bsmallmatrix} ,
    \mu = \begin{bsmallmatrix}
           0 \\
           0 \\
           0
    \end{bsmallmatrix}
    \right),
    \text{where $\initialMap$ denotes the trivial map from empty domain. }
\end{align*}
That is, $s$ takes a length 0 vector and generates a length 3 vector, holding the values of $w, x$ and $y$,
with the normal distribution $\normal(\mu, \cov)$.
This $s$ can be decomposed as in~\Cref{def:display-CI-plain}
with $s_w = (\initialMap, 0, 1)$, $g_x = (1, 0, 1)$, and $g_y = (1, 0, 1)$:
composing $s_2$, $g_x$ and $g_y$ as in~\Cref{def:display-CI-plain},
we get $\EE(w) = \EE(\xi_w) = 0$,
$\EE(x) = \EE(w + \xi_x) = 0 + 0 = 0$ , and similarly,
$\EE(y) = \EE(w + \xi_y) = 0$, justifying the noise's mean $\mu$ being a zero vector.
For the covariance matrix, let $v = (w, x, y) - (\EE(w), \EE(x), \EE(y)) $.
Then
$\cov = \EE(v \cdot v^T) = \EE((w, x, y) \cdot (w, x, y)^T) $, and one may show that $\cov$ is equal to the matrix above.
\end{example}
\begin{restatable}{proposition}{CorCIEq}
\label{cor:proc-markov-DIBI-CI-eq}
    For any $\catWithVarName{\catC}{\obChoice}$-kernel $s \colon
    \empset \to W \cup X \cup Y$ where $W, X, Y$ are mutually disjoint, $\displayCI{X}{Y}{W}$ iff $\dibiCI{X}{Y}{W}$.
\end{restatable}

In order to extend \Cref{cor:proc-markov-DIBI-CI-eq} to the scenario in \Cref{def:DIBI-CI} where a kernel $f$ might contain some $U$ in its codomain that does not appear in the CI statement, we need to modify the notion of CI from \Cref{def:display-CI-plain} -- referred to as plain CI -- to allow objects that do not appear in the CI statement to occur in the codomain of $s$. We suggest two sensible extensions.
%
\begin{definition}
\label{def:proc-markov-CI}
    Given an $\catX$-morphism $s \colon \tenUnit \to \obW \tensor \obX \tensor \obY \tensor \obU$,
    \begin{description}
        \item[\textbf{$s$ displays Markov CI}] if there exist $\catX$-morphisms $s_{\obW}, g_{\obX}, g_{\obY}$ such that \ref{fig:display-CI-markov} holds. We write this as $\markovCI{\obX}{\obY}{\obW}$. \label{item:def-markov-CI}
        \item[\textbf{$s$ displays superset CI}] if there exist $\catX$-morphisms $s_{0}, g_1, g_2$ such that \ref{fig:display-CI-process} holds. We write this as $\procCI{\obX}{\obY}{\obW}$. \label{item:def-process-CI}
    \end{description}
\begin{figure}[bt]
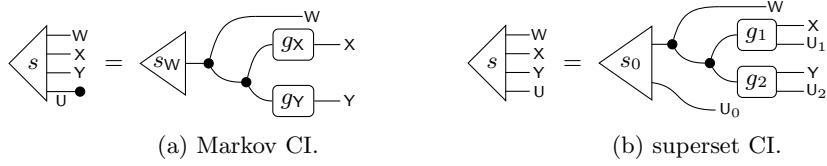

    \begin{subfigure}{0.5\textwidth}
        \input{display-CI-markov-1.tikz} = \input{display-CI-markov-2.tikz}
        \caption{Markov CI.}
        \label{fig:display-CI-markov}
    \end{subfigure}%
    \begin{subfigure}{0.5\textwidth}
        \input{display-CI-process-1.tikz} = \input{display-CI-process-2.tikz}
        \caption{superset CI.}
        \label{fig:display-CI-process}
    \end{subfigure}%
    \caption{Two possible extension of plain CI.}
    \label{fig:display-CI-variation}
\end{figure}%
\end{definition}%
These two notions differ regarding to the treatment of the extra object $\obU$.
In \Cref{fig:display-CI-markov}, we project out the extra object $\obU$ and reduce the situation to that of \Cref{def:display-CI-plain}.
In \Cref{fig:display-CI-process}, $\obU$ is kept and passed along through $s_0, g_1, g_2$.
Clearly, both reduce to \Cref{def:display-CI-plain} when no such $\obU$ appears.
%
We can now state that DIBI CI coincides with Markov CI, but is weaker than superset CI.
\begin{restatable}{theorem}{ThmCIEq}
\label{thm:proc-DIBI-CI-eq}
Given the $\catWithVarName{\catC}{\obChoice}$-kernel $f \colon \empset \to W \cup X \cup Y \cup U$ from \Cref{def:DIBI-CI},
\begin{enumerate}
    \item $f$ satisfies $\markovCI{X}{Y}{W}$ if and only if it satisfies $\dibiCI{X}{Y}{W}$; \label{item:proc-DIBI-CI-eq-1}
    \item if $f$ satisfies $\procCI{X}{Y}{Z}$, then it also satisfies $\dibiCI{X}{Y}{Z}$. \label{item:proc-DIBI-CI-eq-2}
\end{enumerate}
\end{restatable}
The proof of \Cref{item:proc-DIBI-CI-eq-1} is in~\Cref{appendix:sec-CI}.
\Cref{item:proc-DIBI-CI-eq-2} follows from \Cref{item:proc-DIBI-CI-eq-1} and that $\procCI{\obX}{\obY}{\obW}$ implies $\markovCI{\obX}{\obY}{\obW}$: one can simply apply $\discarder_{\obU}$ on both sides of \ref{fig:display-CI-process} and obtain \Cref{fig:display-CI-markov} via naturality if $\discarder$.
The converse of \Cref{item:proc-DIBI-CI-eq-2}
does not hold in general, as demonstrated by the following example.
\renewcommand\windowpagestuff{\[\input{counter-CI-eq-1.tikz}\]}
\opencutright
\begin{example}
\label{ex:different-CI}
    \begin{cutout}{2}{\dimexpr\linewidth-4cm\relax}{0pt}{4}
    Consider the syntactic DIBI model $\lr{\frameSynKernel, \natVal}$ from \Cref{subsec:syn-DIBI-frame}.
    Define the $\catWithVarName{\catSynVar}{\id}$-kernel $f$ as on the right-hand side, where $c_0, c_1, c_2, d$ are all generating morphisms, i.e., not further decomposable.
    Then $f$ satisfies the DIBI CI $\dibiCI{x}{y}{w}$, but not the superset CI $\procCI{x}{y}{w}$: one cannot rewrite the diagram in the dotted box into a juxtaposition of two diagrams with output wires containing $x$ and $y$, respectively; in other words, it cannot be rewritten as the style in \Cref{fig:display-CI-process}.
    \end{cutout}%
\end{example}

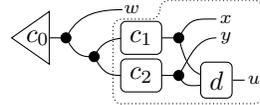
\Cref{ex:different-CI} gives some hint at how to weaken the superset CI to match DIBI CI: there one needs to allow some morphism $d$ following the morphism which witnesses $\procCI{x}{y}{z}$. We formalise this idea and show the resulting notion is indeed equivalent to both Markov and DIBI CI.
%

\vspace*{0.2cm}
\noindent\begin{minipage}{0.5\textwidth}
    \begin{definition}
    \label{def:extended-superset-CI}
    An $\catX$-morphism $s \colon
\tenUnit \to \obW \tensor \obX \tensor \obY \tensor \obU$ \emph{displays the extended superset conditional independence} -- denoted as $\extSupsetCI{\obX}{\obY}{\obW}$ -- if there exist $\catX$-morphisms $s_0, g_1, g_2, h$ such that $s$ can be decomposed as on the right-hand side.
\end{definition}
\end{minipage}%
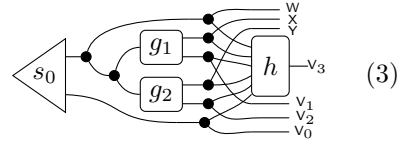
\begin{minipage}{0.5\textwidth}
\begin{equation}
\label{eq:extend-superset-CI}
    \input{extend-superset-CI-2.tikz}
\end{equation}
\end{minipage}
\vspace*{0.2cm}

Compared with \Cref{fig:display-CI-process}, here one allows an extra morphism $h$ to appear after the original superset CI diagrams in \Cref{fig:display-CI-process}; in fact, modulo rewiring, \eqref{eq:extend-superset-CI} is exactly $\scalebox{1}{\input{extend-superset-CI-3.tikz}}$, where $s_1 = \scalebox{.9}{\input{extend-superset-CI-4.tikz}}$.
One intuitive way to think of the extended superset CI is to view the morphisms as certain computational processes~\cite{pavlovic2023programs}: $\obX$ and $\obY$ are independent given $\obW$ in $s$ if $s$ could be obtained via a computation in which $\obX$ and $\obY$ are computed independently from $\obW$ (using $g_1$ and $g_2$ in \eqref{eq:extend-superset-CI} respectively),
after which some further computation may apply (for which stands the $h$ part in \eqref{eq:extend-superset-CI}).
\begin{restatable}{proposition}{PropDiagDibiCI}
\label{prop:extended-superset-CI}
    In Markov categories with conditionals, extended superset CI and Markov CI are equivalent. Therefore, in the context of \Cref{thm:proc-DIBI-CI-eq}, suppose further that $\catC$ has conditionals, then the three notions of CI -- DIBI CI, Markov CI, and extended superset CI -- coincide.
\end{restatable}

\section{Conclusion}
\label{sec:conclusion}
In this paper we provided a general recipe to construct models for DIBI logic, generalising the previously studied probabilistic and relational models. We adopted string diagrams to best visualise the `input-preserving' property that characterises the states in the models, as well as the compositions and subkernel relations, whose definition would be quite convoluted in non-diagrammatic syntax. Then we derived various new classes DIBI models of interest. Also, we abstractly define a notion of conditional independence in terms of DIBI formulas. As our approach is based on Markov categories, we were then able to compare it with other definitions of CI proposed in the literature.

There are many promising directions for future work. On the logic side, DIBI logic -- interpreted in the probabilistic models -- was designed to be the assertion logic of Conditional Probabilistic Separation Logic (CPSL). Our categorical construction of a wide class of DIBI models suggests a generalisation of CPSL to obtain program logics in various scenarios beyond probabilistic programs, in the spirit of Moggi~\cite{moggi1991notions}.

The notion of CI we propose can be straightforwardly generalised from Markov categories to copy-delete categories (see \Cref{sec:preliminaries}). This would allow us to encompass models such as relations with bag semantics in databases~\cite{chaudhuri1993optimization,Green2009bag}, sub-probability measures~\cite{klenke2013probability}. However, to the best of the authors' knowledge, \Cref{cor:proc-markov-DIBI-CI-eq} fails for generic CD categories. Hence, finding appropriate notions of CI in this more general setting remains an open question.

From a categorical perspective, the definition of the category $\catWithVarName{\catC}{\obChoice}$ deserves further exploration, from at least two angles.
First, the $\catWithVarName{\catC}{\obChoice}$-morphisms may be seen as a ``bundle'' of the images of some syntactic categories of variables and renaming (similar to $\catSynVar$ from \Cref{subsec:syn-DIBI-frame}) under suitable functors -- usually referred to as `models' in functorial semantics. We would like to make the connection with functorial semantics rigorous in terms of the categorical structures involved. \cite{lawvere1963functorial,bonchi2018deconstructing}.
Second, while the current work represents finite sets of variables using deduplicated finite $\varLeq$-ordered lists, towards a more principled treatment, it is worth exploring using nominal string diagrams, a diagrammatic calculus for variables and renaming~\cite{balco2018partially,balco2019nominal,balco2020display},
to represent sets of variables.

\bibliographystyle{splncs04}
\bibliography{main}

\begin{thebibliography}{10}
\providecommand{\url}[1]{\texttt{#1}}
\providecommand{\urlprefix}{URL }
\providecommand{\doi}[1]{https://doi.org/#1}

\bibitem{aho1979theory}
Aho, A.V., Beeri, C., Ullman, J.D.: The theory of joins in relational
  databases. ACM Transactions on Database Systems (TODS)  \textbf{4}(3),
  297--314 (1979). \doi{10.1145/320083.320091}

\bibitem{balco2020display}
Balco, S.: Display calculi and nominal string diagrams. Ph.D. thesis,
  University of Leicester (2020)

\bibitem{balco2018partially}
Balco, S., Kurz, A.: Partially monoidal categories and the algebra of
  simultaneous substitutions (2018), available at
  \url{https://gdlyrttnap.pl/resources/papers/syco1.pdf}

\bibitem{balco2019nominal}
Balco, S., Kurz, A.: {Nominal String Diagrams}. In: Roggenbach, M., Sokolova,
  A. (eds.) 8th Conference on Algebra and Coalgebra in Computer Science (CALCO
  2019). Leibniz International Proceedings in Informatics (LIPIcs), vol.~139,
  pp. 18:1--18:20. Schloss Dagstuhl--Leibniz-Zentrum fuer Informatik, Dagstuhl,
  Germany (2019). \doi{10.4230/LIPIcs.CALCO.2019.18}

\bibitem{bao2021bunched}
Bao, J., Docherty, S., Hsu, J., Silva, A.: A bunched logic for conditional
  independence. In: 2021 36th Annual ACM/IEEE Symposium on Logic in Computer
  Science (LICS). pp. 1--14. IEEE (2021). \doi{10.1109/LICS52264.2021.9470712}

\bibitem{barthe2019probabilistic}
Barthe, G., Hsu, J., Liao, K.: A probabilistic separation logic. Proceedings of
  the ACM on Programming Languages  \textbf{4}(POPL),  1--30 (2019).
  \doi{10.1145/3371123}

\bibitem{bonchi2018deconstructing}
Bonchi, F., Soboci{\'n}ski, P., Zanasi, F.: Deconstructing lawvere with
  distributive laws. Journal of logical and algebraic methods in programming
  \textbf{95},  128--146 (2018). \doi{10.1016/j.jlamp.2017.12.002}

\bibitem{chaudhuri1993optimization}
Chaudhuri, S., Vardi, M.Y.: Optimization of real conjunctive queries. In:
  Proceedings of the twelfth ACM SIGACT-SIGMOD-SIGART symposium on Principles
  of database systems. pp. 59--70 (1993). \doi{10.1145/153850.153856}

\bibitem{cho2019disintegration}
Cho, K., Jacobs, B.: Disintegration and {B}ayesian inversion via string
  diagrams. Mathematical Structures in Computer Science  \textbf{29}(7),
  938--971 (2019). \doi{10.1017/S0960129518000488}

\bibitem{dawid1979conditional}
Dawid, A.P.: Conditional independence in statistical theory. Journal of the
  Royal Statistical Society: Series B (Methodological)  \textbf{41}(1),  1--15
  (1979). \doi{10.1111/j.2517-6161.1979.tb01052.x}

\bibitem{docherty2019thesis}
Docherty, S.: Bunched logics: a uniform approach. Ph.D. thesis, UCL (University
  College London) (2019)

\bibitem{fong2019invitation}
Fong, B., Spivak, D.I.: An invitation to applied category theory: seven
  sketches in compositionality. Cambridge University Press (2019).
  \doi{10.1017/9781108668804}

\bibitem{fritz2020markov}
Fritz, T.: A synthetic approach to markov kernels, conditional independence and
  theorems on sufficient statistics. Advances in Mathematics  \textbf{370},
  107239 (Aug 2020). \doi{10.1016/j.aim.2020.107239}

\bibitem{glymour2019review}
Glymour, C., Zhang, K., Spirtes, P.: Review of causal discovery methods based
  on graphical models. Frontiers in genetics  \textbf{10}, ~524 (2019).
  \doi{10.3389/fgene.2019.00524}

\bibitem{golubtsov2002monoidal}
Golubtsov, P.V.: Monoidal kleisli category as a background for information
  transformers theory. Inf. Process.  \textbf{2}(1),  62--84 (2002)

\bibitem{Green2009bag}
Green, T.J.: Bag Semantics, pp. 201--206. Springer US, Boston, MA (2009).
  \doi{10.1007/978-0-387-39940-9_979}

\bibitem{hofmann1995interpretation}
Hofmann, M.: On the interpretation of type theory in locally cartesian closed
  categories. In: Pacholski, L., Tiuryn, J. (eds.) Computer Science Logic. pp.
  427--441. Springer Berlin Heidelberg, Berlin, Heidelberg (1995).
  \doi{10.1007/BFb0022273}

\bibitem{jacobs1994semantics}
Jacobs, B.: Semantics of weakening and contraction. Annals of pure and applied
  logic  \textbf{69}(1),  73--106 (1994). \doi{10.1016/0168-0072(94)90020-5}

\bibitem{klenke2013probability}
Klenke, A.: Probability theory: a comprehensive course. Springer Science \&
  Business Media (2013). \doi{10.1007/978-1-84800-048-3}

\bibitem{lawvere1963functorial}
Lawvere, F.W.: Functorial semantics of algebraic theories. Proceedings of the
  National Academy of Sciences  \textbf{50}(5),  869--872 (1963).
  \doi{10.1073/pnas.50.5.869}

\bibitem{li2023lilac}
Li, J.M., Ahmed, A., Holtzen, S.: Lilac: a modal separation logic for
  conditional probability. Proceedings of the ACM on Programming Languages
  \textbf{7}(PLDI),  148--171 (2023). \doi{10.1145/3591226}

\bibitem{moggi1991notions}
Moggi, E.: Notions of computation and monads. Information and computation
  \textbf{93}(1),  55--92 (1991). \doi{10.1016/0890-5401(91)90052-4}

\bibitem{o1999logic}
O'Hearn, P.W., Pym, D.J.: The logic of bunched implications. Bulletin of
  Symbolic Logic  \textbf{5}(2),  215--244 (1999). \doi{10.2307/421090}

\bibitem{pavlovic2023programs}
Pavlovic, D.: Programs as diagrams: From categorical computability to
  computable categories. Springer Nature (2023).
  \doi{10.1007/978-3-031-34827-3}

\bibitem{pearl2009causality}
Pearl, J.: Causality. Cambridge university press (2009).
  \doi{10.1017/CBO9780511803161}

\bibitem{piedeleuzanasi23}
Piedeleu, R., Zanasi, F.: An introduction to string diagrams for computer
  scientists. CoRR  \textbf{abs/2305.08768} (2023)

\bibitem{pym2013semantics}
Pym, D.J.: The semantics and proof theory of the logic of bunched implications,
  vol.~26. Springer Science \& Business Media (2013).
  \doi{10.1007/978-94-017-0091-7}

\bibitem{reynolds2002separation}
Reynolds, J.C.: Separation logic: A logic for shared mutable data structures.
  In: Proceedings 17th Annual IEEE Symposium on Logic in Computer Science. pp.
  55--74. IEEE (2002), \url{https://dl.acm.org/doi/10.5555/645683.664578}

\bibitem{selinger2010survey}
Selinger, P.: A survey of graphical languages for monoidal categories. In: New
  structures for physics, pp. 289--355. Springer (2010).
  \doi{10.1007/978-3-642-12821-9_4}

\bibitem{tarski1983logic}
Tarski, A.: Logic, semantics, metamathematics: papers from 1923 to 1938.
  Hackett Publishing (1983). \doi{10.2307/2216869}

\end{thebibliography}
%

\appendix
\section{Background on Monads}
\label{appendix:sec-app}
We first recall the basic definition of monads. We refer to \cite[Sect. 3]{fritz2020markov} for an overview of the material in this section.
An endofunctor $\funcT \colon \catC \to \catC$ is a
\emph{monad} if it has a unit $\monadUnit{\funcT}{}: 1_{\catC} \to \funcT$ and
a multiplication $\monadMult{\funcT}{}: \funcT^2 \to \funcT$ natural
transformations satisfying certain compatibility conditions.
%
%
Every monad $\funcT \colon \catC \to \catC$ induces a Kleisli category $\Kl(\funcT)$, whose objects are exactly $\catC$-objects, and morphisms $\obX \to \obY$ are $\catC$ morphisms of type $\obX \to \funcT \obY$, with compositions of $f \colon \obX \to \funcT \obY$ and $g \colon \obY \to \funcT \obZ$ given by $\obX \xrightarrow{f} \funcT \obY \xrightarrow{\funcT g} \funcT \funcT \obZ \xrightarrow{\monadMult{\funcT}{\obZ}} \funcT \obZ$. We will write the morphisms in $\Kl(\funcT)$ as $\obX \klto \obY$ to distinguish them from their counterpart $\obX \to \funcT \obY$ in $\catC$. Importantly, if $\catC$ is a SMC and $\funcT$ is a commutative monad, then $\Kl(\funcT)$ is also an SMC~\cite{jacobs1994semantics}. If $\funcT$ is affine symmetric monoidal, then $\Kl(\funcT)$ is a Markov category~\cite{golubtsov2002monoidal,cho2019disintegration}.

In the remainder of this section, we recall the monads used in this paper: 
the distribution monad $\dist$, 
the powerset monad $\powset$, 
the Giry monad $\monadGiry$,
and the Radon monad $\monadRadon$. 
\begin{definition}[Discrete Distribution Monad]
\label{def:dist-monad}
The discrete distribution monad $\dist$ is an endofunctor on $\catSet$. 
It maps a countable set $X$ 
to the set of distributions over $X$, i.e., the set containing all functions $\mu$ 
over $X$ is satisfying $\sum_{x \in X}\mu(x) = 1$, 
and maps a function $f: X \to Y$ to $\dist(f): \dist(X) \to \dist(Y)$, 
given by $\dist(f)(\mu)(y) \defeq \sum_{f(x) = y} \mu(x)$. 

For the monadic structure, define the unit $\eta$ by $\eta_X(x) \defeq \delta_x$, where $\delta_x$ denotes the Dirac distribution on 
$x$: for any $y \in X$, we
have $\delta_x(y) = 1$ if $y = x$, otherwise $\delta_x(y) = 0$. Further, define $\bind\colon \dist (X) \rightarrow (X \to \dist (Y)) \to \dist(Y)$ by
$\bind(\mu)(f)(y)\defeq \sum_{p \in \dist(Y)} \dist(f)(\mu)(p) \cdot p(y)$.
\end{definition}

\begin{definition}[Powerset monad]
\label{def:powerset-monad}
  The powerset monad $\powset$ is an endofunctor on $\catSet$. 
  It maps every set to the set of its subsets $\powset(X) = \{U \mid U\subseteq X\}$. We
  define $\eta_X\colon X \rightarrow \powset(X)$ mapping each $x \in X$ to the
  singleton $\{x\}$, and $\bind \colon \powset(X) \to (X \to \powset(Y)) \to \powset(Y)$
  by $\bind (U)(f) \defeq \cup \{ y \mid \exists x\in U.  f(x) = y\}$.
\end{definition}
The next monad is defined on the category $\catMeas$ of measurable spaces, which consists of measurable spaces $(A, \sigalg{A} )$ as objects, and measurable functions as morphisms.
$\catMeas$ is a monoidal category, where the monoidal product on objects and morphisms are given by the product of measurable spaces and measurable functions, respectively. In particular, the monoidal unit consists of the singleton measurable space $(\setOne = \{ \bullet \}, \{\empset, \setOne\})$.

\begin{definition}[Giry Monad]
\label{def:giry-monad}
The \emph{giry monad} $\monadGiry$ maps 
a measurable space $(X, \Sigma_{X})$ 
to another measurable space $(\monadGiry(X), \Sigma_{\monadGiry(X)})$, where $\monadGiry(X)$ is the set of probability measures over $X$, and the $\sigma$-algebra $\Sigma_{\monadGiry(X)}$ is the coarsest $\sigma$-algebra over $\monadGiry(X)$ making the evaluation function $\ev_A: \monadGiry(X) \to [0,1]$, defined by $\ev_A(\nu) = \nu(A)$, measurable for any $A \in \Sigma_X$.
For each measurable function $f: X \mapsto Y$,
$\monadGiry f: \monadGiry X \to \monadGiry Y$ is defined by $(\monadGiry f) (\nu)(B) = \nu(f^{-1}(B))$ for $B \in \Sigma_Y$.
For the monadic structure,
	define the unit $\eta$ by $\eta_X(x) = \dirac{x}$;
	define the bind operator $\dbind_{X, Y}: \monadGiry X \to ((X \to \monadGiry Y) \to \monadGiry Y)$ by
	$\dbind(\nu)(f)(B) = \int_X f(X)(B) d \nu$ for $B \in \Sigma_{\monadGiry Y}$.
\end{definition}
\begin{definition}[Radon Monad]
\label{def:radon-monad}
The Radon monad $\monadRadon$ is a measure monad on the category of compact Hausdorff spaces. If maps a compact Hausdorff space $X$ to the set of Radom measures $\mu$ on $X$ such that $\mu(X) \leq 1$. It maps a continuous map between compact Hausdorff spaces $f: X \to Y$ to the push-forward measure $\monadRadon(f): \monadRadon X \to \monadRadon Y$ given by $\dist(f)(\mu)(y) \defeq  \mu(f^{-1} (y))$. 

For the monadic structure:
we define the unit $\eta$ to take a point $x \in X$ to the diract distribution $\delta_x$ solely supported at $x$. 
We also define the bind operator $\dbind_{X, Y}: \monadRadon X \to ((X \to \monadRadon Y) \to \monadRadon Y)$ by
	$\dbind(\nu)(f)(B) = \int_X f(X)(B) d \nu$.
\end{definition}

The category of stochastic processses $\catStoch$ is the Kleisli category of the Giry monad $\monadGiry$. It is helpful to explicate its structure.
\begin{definition}
    The symmetric monoidal category of stochastic processses $\catStoch$ has the following components: 
    \begin{itemize}
        \item objects are measurable spaces $(A, \sigalg{A} )$; 
        \item morphisms $(A, \sigalg{A}) \to (B, \sigalg{B})$ are maps $f \colon \sigalg{B} \times A \to [0, 1]$ satisfying: for arbitrary $T \in \sigalg{B}$, $f(T, -) \colon A \to [0, 1]$ is measurable, and for arbitrary $a \in A$, $f(-, a) \colon \sigalg{B} \to [0, 1]$ is a probability measure; 
        \item compositions of $f \colon (A, \sigalg{A}) \to (B, \sigalg{B})$ and $g \colon (B, \sigalg{B}) \to (C, \sigalg{C})$ is the map $g \after f \colon \sigalg{C} \times A \to [0, 1]$ mapping $(U, a)$ to $\int_{b \in B} g(U, b) \cdot f(d b, a)$;  
        \item the identity morphism $\idMor$ on $(A, \sigalg{A})$ maps $(S, a) \in \sigalg{A} \times A$ to $1$ if $a \in S$, and to $0$ if $a \not \in S$; 
        \item the monoidal product $\tensor$ acts on objects as $(A, \sigalg{A}) \tensor (B, \sigalg{B}) = (A \times B, \sigalg{A} \tensor \sigalg{B})$, where $\sigalg{A} \tensor \sigalg{B}$ is the smallest sigma-algebra containing $\sigalg{A} \times \sigalg{B})$; 
        \item the monoidal product $\tensor$ acts on morphisms to obtain product measures. That is, $(U, V) \in \sigalg{B} \times \sigalg{D}$ as follows: given $f \colon (A, \sigalg{A}) \to (B, \sigalg{B})$ and $g \colon (C, \sigalg{C}) \to (D, \sigalg{D})$, $f \tensor g \colon \sigalg{B} \tensor \sigalg{D} \times A \times C \to [0, 1]$ maps $(U, V, a, c)$ to $f(U, a) \cdot g(V, b)$.
    \end{itemize}
\end{definition}

\section{Omitted Proofs from \Cref{sec:Markov-DIBI-model}}
\label{sec:omitted-proofs-DIBI-model}
This section contains the missing proof of statements in \Cref{sec:Markov-DIBI-model}, as well as some useful properties of the $\catWithVarName{\catC}{\obChoice}$-kernels and the DIBI model. We stay with the setting in \Cref{sec:Markov-DIBI-model} for $\catC$, $\setVar$, and $\obChoice$.

\begin{proposition}
\label{prop:markov-preserved}
    If $\catC$ is a Markov category, then $\catWithVarName{\catC}{\obChoice}$ is also Markovian.
    If Markov cateogry $\catC$ has conditionals, then so does $\catWithVarName{\catC}{\obChoice}$.
\end{proposition}
\begin{proof}
    Both follow immediately from the construction of $\catWithVarName{\catC}{\obChoice}$: note that $\catWithVarName{\catC}{\obChoice}$-morphisms are $\catC$-morphisms.
\end{proof}

The following observation says that, in $\catWithVarName{\catC}{\obChoice}$-kernels, if one forgets the new variables in the output, then one get exactly identity on the input variables.
\begin{proposition}
\label{prop:input-preserving}
    For an arbitrary $\catWithVarName{\catC}{\obChoice}$-kernel $f \colon X \to Y$, $(\idMor_{X} \tensor \discardMor_{Y\setminus X}) \after f = \idMor_{X}$.
\end{proposition}
\begin{proof}
    This is an immediate consequence of the Markovian property:
    \begin{align*}
        \input{input-preserving-2.tikz} = \input{input-preserving-4.tikz} = \input{input-preserving-3.tikz}
    \end{align*}
    where we assume $X = \{ x_1, \dots, x_m \}$, $Y \setminus X = \{ u_1, \dots, u_k \}$.
\end{proof}
The class of $\catWithVarName{\catC}{\obChoice}$-kernels is closed under both parallel and sequential compositions.%
\begin{proposition}\label{prop:kernel-closed-under-composition}
For arbitrary $\catWithVarName{\catC}{\obChoice}$-kernels $f \colon X \to Y$ and $g \colon U \to V$, whenever $f \operateVar g$ is defined, the result $f \operateVar g$ is also an $\catWithVarName{\catC}{\obChoice}$-kernel, for $\operateVar \in \{ \stPlus, \stThen \}$.
\end{proposition}
\begin{proof}
Suppose $f \stThen g$ is defined, then spelling out the definition, $Y = U$, 
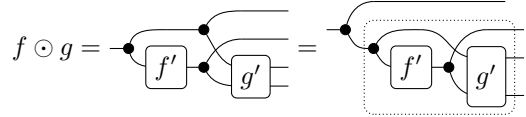
\[
    f \stThen g = \input{seq-comp-closed-1.tikz} = \input{seq-comp-closed-2.tikz}
\]
therefore $f \stThen g$ is an input-preserving kernel.

Suppose $f \stPlus g$ is defined. That is, $X \cap U = Y \cap V$. Follow the notation in \Cref{def:kernel-composition}, $f \stPlus g$ is also an input-preserving kernel:
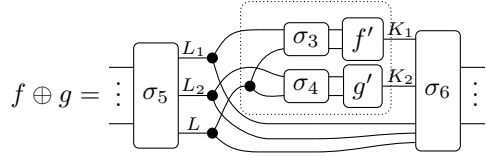
\[
    f \stPlus g = \input{para-comp-closed-2.tikz}
\]
The morphisms inside the dotted square play the role of the nontrivial part of the input-preserving kernel.
\end{proof}
\begin{proposition}
\label{prop:subkernel-preorder}
    The subkernel relation $\stLeq$ on $\catWithVarName{\catC}{\obChoice}$-kernels (\Cref{def:kernel-leq}) is a preorder.
\end{proposition}
\begin{proof}
    In other words, we prove $\stLeq$ is reflexive and transitive.
    The diagram in \Cref{def:kernel-leq} is trivial when spelled out for witnessing $f \stLeq f$.
    Suppose $f_1 \stLeq f_2$ and $f_2 \stLeq f_3$, then they are witnessed by:
    \begin{align*}
        f_2 & = \input{stLeq-pre-order-0.tikz}
        & f_3 & = \input{stLeq-pre-order-1.tikz}
    \end{align*}
    Then $f_1 \stLeq f_3$ is witnessed by the following reasoning:
    \[
        f_3 = \input{stLeq-pre-order-1.tikz} = \input{stLeq-pre-order-2.tikz} = \input{stLeq-pre-order-3.tikz}
    \]
    where rewiring $\sigma$ and $\catWithVarName{\catC}{\obChoice}$-kernel $g$ are obtained by the morphisms in the two dotted areas, respectively.
\end{proof}
We are now ready to provide a proof of Proposition~\ref{prop:subkernel-unique}. We first restate the proposition. 
\PropSubkernelUnique*
\begin{proof}[Proof of Proposition~\ref{prop:subkernel-unique}]
Suppose $f_i \stLeq g$ is witnessed by the following diagrams, where $i = 1, 2$:
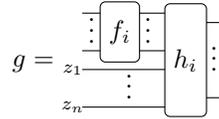
\[
    g = \input{subkernel-unique-1.tikz}
\]
where we don't worry about the rewiring morphisms, and $h_1, h_2$ are some $\catWithVarName{\catC}{\obChoice}$-kernels. Discarding the $z_i$s via $\discardMor{z_i}$, we obtain 
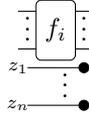
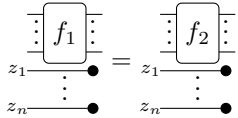
\[\input{subkernel-unique-2.tikz}\]
Therefore, 
\[\input{subkernel-unique-3.tikz} = \input{subkernel-unique-4.tikz}\]
By the assumption on $\catC$, this entails $f_1 = f_2$.
\end{proof}

Next, we prove the central result \Cref{thm:inst-DIBI}.
We first present two lemmas that are useful in verifying the DIBI frame conditions.
\begin{lemma}
\label{lem:syn-kernel-plus-id-morphism}
For any $\catWithVarName{\catC}{\obChoice}$-kernel $f \colon X \to Y$, and any $U \fsubseteq \setVar$ such that $U \cap X = U \cap Y = \empset$, by viewing $\idMor_{U}$ as a (trivial) $\catWithVarName{\catC}{\obChoice}$-kernel, $f \stPlus \idMor_{U}$ is defined, and $f \stPlus \idMor_{U} = \sigma_2 \after (f \tensor \idMor_{U}) \after \sigma_1$, for some rewiring $\sigma_1$ and $\sigma_2$.
\end{lemma}
\begin{proof}
That $f \stPlus \idMor_{U}$ is defined follows immediately from the assumption that $X \cap U = Y \cap U ( = \empset)$.
Then $f \stPlus \idMor_{U} = \sigma_2 \after (f \tensor \idMor_{U}) \after \sigma_1$ follows immediately by the definition of parallel composition.
\end{proof}
\begin{lemma}
\label{lem:subkernel-using-plus}
Two $\catWithVarName{\catC}{\obChoice}$-kernels $f$ and $g$ satisfy $f \stLeq g$ if and only if there exist a set of variables $U$ and another $\catWithVarName{\catC}{\obChoice}$-kernel $h$ such that $g = (f \stPlus \idMor_{U}) \stThen h$.
\end{lemma}
\begin{proof}
    This follows immediately from \Cref{lem:syn-kernel-plus-id-morphism}, by noticing that taking the parallel composition with $\idMor_{U}$ -- when viewed as a trivial kernel -- is nothing but the monoidal product with $\idMor_{U}$, modulo some rewiring.
\end{proof}

Now we are ready to check all the DIBI frame conditions for the syntactic DIBI frame whose states are input-preserving diagrams in the freely generated Markov category of string diagrams. We first restate Theorem~\ref{thm:inst-DIBI}.
%
\ThmInstDibi*
\begin{proof}[Proof of Theorem~\ref{thm:inst-DIBI}]
We verify all the frame conditions in \Cref{fig:dibi-frame} for the $\catWithVarName{\catC}{\obChoice}$-kernels.
We omit the references to~\Cref{lem:syn-kernel-plus-id-morphism} and~\Cref{lem:subkernel-using-plus} when using them.
\begin{enumerate}
    \item \eqref{frameAx:plus-com}. Note that the definition of $\stPlus$ over $\catWithVarName{\catC}{\obChoice}$-kernels does not rely on the specific order of the two kernels, thus commutativity holds; in other words, commutativity of $\stPlus$ holds by definition.

    \item \eqref{frameAx:plus-unit-exist}.
    Given an arbitrary kernel $f$, show there exists a kernel $e$ such that $f \stPlus e = f$ holds. In fact, let $e = \idMor_{\emptyList}$ (where $\emptyList$ is the empty list), which is trivially input-preserving, hence is a $\catWithVarName{\catC}{\obChoice}$-kernel. $f \stPlus e$ is defined: $\domain{f} \cap \domain{\idMor_{\emptyList}} = \empset = \codom{f} \cap \codom{\idMor_{\emptyList}}$.
    Moreover, $f \stPlus \idMor_{\emptyList} = f$.

    \item \eqref{frameAx:plus-ass}.
    Given $\catWithVarName{\catC}{\obChoice}$-kernels $f, g, h$, suppose $f \stPlus g$ and $(f \stPlus g) \stPlus h$ are both defined. We show that $g \stPlus h$ and $f \stPlus (g \stPlus h)$ are also defined; moreover, $f \stPlus (g \stPlus h) = (f \stPlus g) \stPlus h$. The converse holds by a similar argument, hence omitted.

    Let's first show that they are defined. By definition, to verify $g \stPlus h$ is defined amounts to showing $\domain{g} \cap \domain{h} = \codom{g} \cap \codom{h}$. Note that, by $f \stPlus g$ is defined, we have $\codom{g} = \codom{f \stPlus g} \setminus (\codom{f} \setminus (\codom{f} \cap \codom{g}))$. Therefore we can calculate $\codom{g} \cap \codom{h}$ as: \allowdisplaybreaks
    \begin{linenomath*}
    \begin{align*}
        \codom{g} \cap \codom{h}
        & = \left( \codom{f \stPlus g} \setminus \left( \codom{f} \setminus \left( \codom{f} \cap \codom{g} \right) \right) \right) \cap \codom{h} \\
        & = \left( \codom{f \stPlus g} \cap \codom{h} \right) \setminus \left( \codom{f} \setminus \left( \codom{f} \cap \codom{g} \right) \right) \\
        & = \left( \domain{f \stPlus g} \cap \domain{h} \right) \setminus \left( \codom{f} \setminus \left( \codom{f} \cap \codom{g} \right) \right) \\
        & = \left( \domain{f \stPlus g} \setminus \left( \codom{f} \setminus \left( \codom{f} \cap \codom{g} \right) \right) \right) \cap \domain{h} \\
        & = \left( \domain{f \stPlus g} \setminus \left( \codom{f} \setminus \left( \domain{f} \cap \domain{g} \right) \right) \right) \cap \domain{h} \\
        & = \left( \left( \domain{f \stPlus g} \setminus \codom{f} \right) \cup \left( \domain{f \stPlus g} \cap (\domain{f} \cap \domain{g}) \right) \right) \cap \domain{h} \\
        & = \left( \left( \domain{f \stPlus g} \setminus \domain{f} \right)  \cup \left( \domain{f} \cap \domain{g} \right) \right) \cap \domain{h} \\
        & = \left( \left( \left( \domain{f} \cup \domain{g} \right) \setminus \domain{f} \right)  \cup \left( \domain{f} \cap \domain{g} \right) \right) \cap \domain{h} \\
        & = \left( \left( \domain{g} \setminus \domain{f} \right)  \cup \left( \domain{f} \cap \domain{g} \right) \right) \cap \domain{h} \\
        & = \domain{g} \cap \domain{h}
    \end{align*}
    \end{linenomath*}
    
    To show that $f \stPlus (g \stPlus h)$ is defined, we first note that a similar argument as that above shows $\domain{f} \cap \domain{h} = \codom{f} \cap \codom{h}$. Then, $\domain{f} \cap \domain{g \stPlus h} = \codom{f} \cap \codom{g \stPlus h}$ follows immediately:
    
    \begin{linenomath*}
    \begin{align*}
        \codom{f} \cap \codom{g \stPlus h}
        & = \codom{f} \cap (\codom{g} \cup \codom{h}) \\
        & = \left( \codom{f} \cap \codom{g} \right) \cup \left(\codom{f} \cap \codom{h} \right) \\
        & = \left( \domain{f} \cap \domain{g} \right) \cup \left( \domain{f} \cap \domain{h} \right) \\
        & = \domain{f} \cap \left( \domain{g} \cup \domain{h} \right) \\
        & = \domain{f} \cap \domain{g \stPlus h}
    \end{align*}
    \end{linenomath*}
    The equivalence of $(f \stPlus g) \stPlus h$ and $f \stPlus (g \stPlus h)$ follows immediately from their diagrammatic presentation.

    \item \eqref{frameAx:then-unit-left}.
	We show that for arbitrary $\catWithVarName{\catC}{\obChoice}$-kernel $f$, there exists a $\catWithVarName{\catC}{\obChoice}$-kernel $e_{L}$ such that $e_{L} \stThen f = f$.
    Suppose $\domain{f} = \listx$, then simply define $e_{L} \coloneqq \idMor_{\listx}$, which is also a $\catWithVarName{\catC}{\obChoice}$-kernel kernel; moreover, it satisfies $e_{L} \stThen f = f \after \idMor_{\listx}  = f$.

	\item \eqref{frameAx:then-unit-right}.
    Similar to the \eqref{frameAx:then-unit-left} case. Suppose $\codom{f} = \listy$, then let $e_{R} = \idMor_{\listy}$. $e_{R}$ satisfies $f \stThen e_{R} = \idMor_{\listy} \after f = f$.

    \item \eqref{frameAx:then-ass}.
    Note that the sequential operator $\stThen$ on $\catWithVarName{\catC}{\obChoice}$-kernels is exactly the sequential composition in the category $\catWithVarName{\catC}{\obChoice}$, which is associative by definition.

    \item \eqref{frameAx:plus-unit-coherent}.
    For arbitrary $\catWithVarName{\catC}{\obChoice}$-kernels $f$ and $g$ such that $f \stPlus g$ is defined, we show $f \stPlus g \stGeq f$.
    Following the assumption in \Cref{def:kernel-composition}, $f \stPlus g$ is of the following form, where we omit the rewiring for simplicity:
    \[
        f \stPlus g = \input{def-kernel-para-comp-6.tikz} = \input{def-kernel-para-comp-7.tikz}
    \]
    The two parts in dotted boxes are $f$ and $\idMor_{\codom{f}} \tensor g$ (which is also a $\catWithVarName{\catC}{\obChoice}$-kernel), respectively. This witnesses $(f\stPlus g) \stGeq f$. 
    
	\item \eqref{frameAx:then-coherence-right}. We show that for arbitrary $\catWithVarName{\catC}{\obChoice}$-kernels $f, g$, if $f \stThen g$ is defined, then $f \stThen g \stGeq f$. Given the assumption, $f \stThen g = (f \tensor \idMor_{\emptyList}) \stThen g$, which witnesses that $f \stThen g \stGeq f$.

    \item \eqref{frameAx:unit-closure}. Since the set $E$ is the set of all syntactic input-preserving kernels in the current DIBI frame, this condition is trivially satisfied.
    
    \item \eqref{frameAx:plus-down-closed}. Suppose for two $\catWithVarName{\catC}{\obChoice}$-kernels of the form $f \colon X \to Y$ and $g \colon U \to V$, $f \stPlus g$ is defined, and there are two subkernels $f_1 \stSmaller f$, $g_1 \stSmaller g$, where $f_1 \colon X_1 \to Y_1$ and $g_1 \colon U_1 \to V_1$. The goal is to that $f_1 \stPlus g_1$ is defined, and is a subkernel of $f \stPlus g$.
    To see $f_1 \stPlus g_1$ is defined, since $\codom{f_1} \cap \domain{f} = \domain{f_1}$ and $\codom{g_1} \cap \domain{g} = \domain{g_1}$, we have:
    \begin{linenomath*}
    \begin{align*}
        \domain{f_1} \cap \domain{g_1} & = (\codom{f_1} \cap \domain{f}) \cap (\codom{g_1} \cap \domain{g}) \\
        & = (\codom{f_1} \cap \codom{g_1}) \cap (\domain{f} \cap \domain{g}) \\
        & = (\codom{f_1} \cap \codom{g_1}) \cap (\codom{f} \cap \codom{g}) \\
        & = (\codom{f_1} \cap \codom{f}) \cap (\codom{g_1} \cap \codom{g}) \\
        & = \codom{f_1} \cap \codom{g_1}
    \end{align*}
    \end{linenomath*}
    Next, we show the subkernel relation $f_1 \stPlus g_1 \stLeq f \stPlus g$. By Lemma~\ref{lem:subkernel-using-plus}, we can assume that $f = (f_1 \stPlus \idMor_{S}) \stThen f_2$, $g = (g_1 \stPlus \idMor_{T}) \stThen g_2$, where $f_2, g_2$ are also $\catWithVarName{\catC}{\obChoice}$-kernels. Then, diagrammatically we have:
    \[
        f \stPlus g = \left( \input{para-down-closed-1.tikz} \right) \stPlus \left( \input{para-down-closed-2.tikz} \right) = \input{para-down-closed-3.tikz}
    \]
    Notice that the diagrams in the dotted circle is precisely $f_1 \stPlus g_1$, therefore $f_1 \stPlus g_1 \stLeq f \stPlus g$.
    
    We can also derive the desired property using some other frame conditions as follows:

    \begin{linenomath*}
		\begin{align*}
		  f \stPlus g
			&= ((f_1 \stPlus \idMor_{S}) \stThen f_2) \stPlus ((g_1 \stPlus \idMor_{T}) \stThen g_2)  & \\
			&= ( (f_1 \stPlus \idMor_{S}) \stPlus (g_1 \stPlus \idMor_{T}) ) \stThen (f_2 \stPlus g_2) & \eqref{frameAx:rev-exchange} \\ 
			&= (f_1 \stPlus g_1) \stPlus (\idMor_{S} \stPlus \idMor_{T})) \stThen (f_2 \stPlus g_2) & \eqref{frameAx:plus-ass}, \eqref{frameAx:plus-com} \\ 
			&= ((f_1 \stPlus g_1) \stPlus ( \id_{(S \cup T)})  \stThen (f_2 \stPlus g_2) \\
			& \stGeq f_1 \stPlus g_1 & (\text{Def. \ref{def:kernel-leq}})
	\end{align*}
    \end{linenomath*}
    Crucially, the proof below of \eqref{frameAx:rev-exchange}, \eqref{frameAx:plus-ass}, and \eqref{frameAx:plus-com} does not rely on \eqref{frameAx:plus-down-closed}.

	\item \eqref{frameAx:then-up-closed}. Given kernels $f_1, g_1, h$ such that $f_1 \stThen g_1$ is defined and $f_1 \stThen g_1 \stLeq h$, we show that there exist $f_2 \stGeq f_1$, $g_2 \stGeq g_1$, such that $f_2 \stThen g_2$ is well-defined and is exactly $h$.
    By definition, $f_1 \stThen g_1 \stLeq h$ means that there exist a set of variables $U$ and a kernel $h_1$ such that:
    \[
        h = \input{seq-up-closed-1.tikz}
    \]
    We simply define
    \[
        f_2 = \input{seq-up-closed-2.tikz}
        \qquad g_2 = \input{seq-up-closed-3.tikz}.
    \]
    It is obvious then that $f_2 \stGeq f_1$ and $g_2 \stGeq g_1$; moreover, $f_2 \stThen g_2 = h$.

    \item \eqref{frameAx:rev-exchange}.
    For arbitrary $\catWithVarName{\catC}{\obChoice}$-kernels $f_1, f_2, g_1, g_2$, if $(f_1 \stThen f_2) \stPlus (g_1 \stThen g_2)$ is defined, then $f_1 \stPlus g_1$ and $f_2 \stPlus g_2$ are defined as well, and $(f_1 \stThen f_2) \stPlus (g_1 \stThen g_2) = (f_1 \stPlus g_1) \stThen (f_2 \stPlus g_2)$.
    
    We first verify that $f_1 \stPlus g_1$ and $f_2 \stPlus g_2$ are defined. For $f_1 \stPlus g_1$, we need to show that $\domain{f_1} \cap \domain{g_1} = \codom{f_1} \cap \codom{g_1}$. Half of the equation is `free' given the `non-decreasing' nature of the kernels, namely $\domain{f_1} \cap \domain{g_1} \subseteq \codom{f_1} \cap \codom{g_1}$. So it suffices to prove the other inclusion.
    Note that $\domain{f_1} \cap \domain{g_1} = \domain{f_1 \stThen f_2} \cap \domain{g_1 \stThen g_2}$, so it suffices to show that $\codom{f_1} \cap \codom{g_1} \subseteq \domain{f_1 \stThen f_2} \cap \domain{g_1 \stThen g_2}$, as follows:
    \begin{linenomath*}
    \begin{align*}
        \codom{f_1} \cap \codom{g_1} & = \domain{f_2} \cap \domain{g_2} & (\text{Definition of } \stThen) \\
        & \subseteq \codom{f_2} \cap \codom{g_2} & (\text{Definition of kernels}) \\
        & = \codom{f_1 \stThen f_2} \cap \codom{g_1 \stThen g_2} & (\text{Definition of } \stThen) \\
        & = \domain{f_1 \stThen f_2} \cap \domain{g_1 \stThen g_2}
        & ( (f_1 \stThen f_2) \stPlus (g_1 \stThen g_2) \text{ is defined} )
    \end{align*}
    \end{linenomath*}
    A similar argument confirms that $f_2 \stPlus g_2$ is defined.

    Next we show that $(f_1 \stPlus g_1) \stThen (f_2 \stPlus g_2)$ is defined, namely $\codom{f_1 \stPlus g_1} = \domain{f_2 \stPlus g_2}$, as follows:

    \begin{linenomath*}
    \begin{align*}
        \codom{f_1 \stPlus g_1} & = \codom{f_1} \cup \codom{g_1} \\
        & = \domain{f_2} \cup \domain{g_2} \\
        & = \domain{f_2 \stPlus g_2}
    \end{align*}
    \end{linenomath*}

    Finally, for the equivalence $(f_1 \stPlus g_1) \stThen (f_2 \stPlus g_2) = (f_1 \stThen f_2) \stPlus (g_1 \stThen g_2)$ we draw their diagrams:
    \[
        f_1 \stThen f_2 = \input{rev-exch-2.tikz}
        \qquad 
        g_1 \stThen g_2 = \input{rev-exch-3.tikz}
    \]
    \[
        (f_1 \stThen f_2) \stPlus (g_1 \stThen g_2) = \input{rev-exch-4.tikz}
    \]
    where $\domain{f_1} = X_1 \cup U$, $\domain{g_1} = Y_1 \cup U$, and $\domain{f_1} \cap \domain{g_1} = U$.
    This is exactly $(f_1 \stPlus g_1) \stThen (f_2 \stPlus g_2)$, as shown in the diagram below, where the two diagrams inside the dotted circles are $f_1 \stPlus g_1$ and $f_2 \stPlus g_2$, respectively:
    \[
        (f_1 \stPlus g_1) \stThen (f_2 \stPlus g_2) = \input{rev-exch-1.tikz}
    \]
\end{enumerate}
This completes the verification that $\lr{\setSemKernel{\catWithVarName{\catC}{\obChoice}}, \stLeq, \stPlus, \stThen, \setSemKernel{\catWithVarName{\catC}{\obChoice}}}$ satisfies all DIBI frame conditions, and therefore it is a DIBI frame.
\end{proof}

\section{Omitted Proofs from \Cref{sec:applications}}
\label{app:proof-application}
We recall \Cref{prop:prob-KlD-model-eq}, which states the isomorphism between the probabilistic DIBI frames and the categorical DIBI frames based on $\Kl(\dist)$.
\PropProbModelEq*
\begin{proof}[Proof of \Cref{prop:prob-KlD-model-eq}]
    Recall the embedding $\colon \catMemProb \to \catWithVarName{\Kl(\dist)}{\obChoice}$ from \Cref{ex:cat-with-var-name}. We show that the restriction of such $\iota$ to probabilistic kernels -- a subclass of $\catMemProb$-objects -- which we also write $\iota$ with an abuse of notation witness the desired isomorphism of DIBI frames. 
    First, $\iota$ establish a bijection between probabilistic kernels of the form $\memSpaceS{X} \to \dist \memSpaceS{Y}$ and $\catWithVarName{\Kl(\dist)}{\setVal}$-kernel of the form $\setToList{X} \to \setToList{Y}$ -- recall that $\setToList{\cdot}$ denotes the list representation of finite sets of variables. For simplicity we assume $X = \{ x_1, \dots, x_m \}$, $Y = \{ y_1, \dots, y_n \}$.
    Given a probabilistic kernel $f \colon \memSpaceS{X} \to \dist \memSpaceS{Y}$, its image $\iota(f) \colon [x_1, \dots, x_m] \to [y_1, \dots, y_n]$ is a $\Kl(\dist)$-morphism $\setVal^{m} \klto \setVal^{n}$ obtained by the composition: 
    \[
        \setVal^{m} \xrightarrow{\cong} \memSpaceS{X} \xrightarrow{f} \dist \memSpaceS{Y} \xrightarrow{\dist \cong} \dist \setVal^n
    \]
    where $\cong$ is the isomorphism $\memSpaceS{Y} \xrightarrow{\cong} \memSpaceS{y_1} \times \cdots \times \memSpaceS{y_n} \xrightarrow{\cong} \setVal^n$. It satisfies \Cref{def:cat-IP-kernel} immediately by the input-preserving conditions (see \Cref{def:probabilistic-kernel}) of the probabilistic kernel $f$.

    Before verifying that $\iota$ respects the structure on DIBI frames, for convenience we introduce the notion of \emph{combination} of memories: two memories $\memoryM \in \memSpaceS{X}$ and $\memoryN \in \memSpaceS{Y}$ are combinable if $X \cap Y = \empset$; in this case, their combination $\memoryM \memComb \memoryN$ is a memory in $\memSpaceS{X \cup Y}$, such that 
    \[
        \memoryM \memComb \memoryN \colon (u \in X \cup Y) \mapsto
        \begin{cases}
            \memoryM(u) & u \in X \\
            \memoryN(u) & v \in Y
        \end{cases}
    \]
    Now we show $\iota$ respects both compositions. The sequential composition in probabilistic kernels is obviously the Kleisli composition, as already observed in~\cite{bao2021bunched}.
    So we focus on the parallel composition. Consider the $\catWithVarName{\Kl(\dist)}{\setVal}$-kernels $f, g$ from \Cref{def:kernel-composition}. Their counterpart probabilistic kernels $\iota^{-1}(f) \colon \memSpaceS{X} \to \dist \memSpaceS{Y}$ and $\iota^{-1}(g) \colon \memSpaceS{U} \to \dist \memSpaceS{V}$ does the following:
    \begin{align*}
        \iota^{-1}(f) & \colon (\memoryM \in \memSpaceS{X}) \mapsto \sum_{\memoryN \in \memSpaceS{Y\setminus X}} \iota^{-1}(f')(\memoryM)(\memoryN) \ket{\memoryM \memComb \memoryN} \\
        \iota^{-1}(g) & \colon (\memoryM \in \memSpaceS{U}) \mapsto \sum_{\memoryN \in \memSpaceS{V \setminus U}} \iota^{-1}(g')(\memoryM)(\memoryN) \ket{\memoryM \memComb \memoryN}
    \end{align*}
    where $f' \colon X \to (Y\setminus X)$ and $ g' \colon U \to (V \setminus U)$ are the nontrivial parts of the kernels $f$ and $g$, respectively (see \Cref{def:cat-IP-kernel}). The parallel composition of these two probabilistic kernels is as follows, according to \Cref{def:probabilistic-kernel}. Given $\memoryM \in \memSpaceS{X \cup U}$, $\memoryN \in \memSpaceS{Y \cup V}$,
    \begin{align*}
        \iota^{-1}(f) \stPlus \iota^{-1}(g) (\memoryM)(\memoryN) & = \iota^{-1}(f)(\memoryM^{X})(\memoryN^{Y}) \cdot 
        \iota^{-1}(g)(\memoryM^{U})(\memoryN^{V}) \\
        & = \iota^{-1}(f')(\memoryM^{X})(\memoryN^{Y\setminus X}) \cdot 
        \iota^{-1}(g')(\memoryM^{U})(\memoryN^{V \setminus U})
    \end{align*}
    The probabilistic kernel counterpart of $f \stPlus g$ is:
    \begin{align*}
        \iota^{-1}(f \stPlus g) \colon & (\memoryM \in \memSpaceS{X \cup U}) \\
        & \mapsto \sum_{\memoryL_1 \in \memSpaceS{Y \setminus X}, \memoryL_2 \in \memSpaceS{V \setminus U}} \iota^{-1}(f')(\memoryM^{X})(\memoryL_1) \cdot \iota^{-1}(g')(\memoryM^{U})(\memoryL_2) \ket{\memoryM \memComb \memoryL_1 \memComb \memoryL_2}
    \end{align*}
    That is, for arbitrary $\memoryM \in \memSpaceS{X \cup U}$ and $\memoryN \in \memSpaceS{Y \cup V}$,
    \[
    \iota^{-1}(f \stPlus g)(\memoryM)(\memoryN) = \iota^{-1}(f')(\memoryM^{X})(\memoryN^{Y\setminus X}) \cdot \iota^{-1}(g')(\memoryM^{U})(\memoryN^{V \setminus U})
    \]
    Therefore $\iota^{-1}(f) \stPlus \iota^{-1}(g) = \iota^{-1}(f \stPlus g)$. 

    Finally, as the subkernel relation is defined in terms of the sequential and parallel compositions in the same way for both the probabilistic kernels and $\catWithVarName{\Kl(\dist)}{\setVal}$ kernels, $\iota$ also respects the subkernel relation. Therefore we can conclude that $\probFrameBasedOn{\setVal}$ and $\frameSemKernel{\catWithVarName{\Kl(\dist)}{\setVal}}$ are isomorphic as DIBI frames.
\end{proof}

\section{Partially Monoidal Categories}
\label{sec:partially-monoidal-categories}
In this section we prove that in the DIBI frames constructed via our recipe, the parallel composition forms a partially monoidal structure. We formalise the result using the notion of partially monoidal categories~\cite{balco2018partially}.

A \emph{partial monoid} $\lr{A, e, \diamond, D}$ consists of a domain $D \subseteq A \times A$, a binary operation $\diamond \colon D \to A$ whose (partial) unit is $e \in D$, such that the monoidal equations hold whenever defined: $(a \diamond b) \diamond c \eqWhenDefined a \diamond (b \diamond c)$ and $e \diamond a = a = a \diamond e$, where $\eqWhenDefined$ stands for `equal when either side is defined'. 
%
\begin{definition}[\cite{balco2018partially}] 
    A \emph{(strict) partially monoidal category (PMC)} consists of: 
    \begin{itemize}
        \item a small category $\catC$ with sets of objects $\ob{\catC}$ and sets of morphisms $\mor{\catC}$;  
        \item partial monoids $\lr{\ob{\catC}, \obE_0, \tensor_0, D_0}$ and $\lr{\mor{\catC}, \obE_1, \tensor_1, D_1}$, such that $\catD$ with objects from $D_0$ and morphisms from $D_1$ is a subcategory of $\catC \times \catC$; 
        \item the operator $\tensor$ defined as $\tensor_0$ on objects and $\tensor_1$ on morphisms forms a functor $\catD \to \catC$.
    \end{itemize}
\end{definition}
Intuitively, a PMC is a category with compatible partial monoid structures on both the sets of objects and morphisms.

We familiarise the notion of PMC with the probabilistic DIBI models.

\begin{restatable}{proposition}{PropProbFrPMC}
\label{prop:prob-frame-PMC}
    $\lr{ \catProbKern, \stPlus, \idMor_{\memSpaceS{\empset}} }$ is a partially monoidal category, where $\catProbKern$ is the category of probabilistic kernels viewed as a subcategory of $\Kl(\dist)$.
\end{restatable}
\begin{proof}
    Let us define the partial monoids on the set of objects and morphisms, respectively.

    On objects, the partial monoid $\lr{\ob{\catProbKern}, \obE_0, \tensor_0, D_0}$ is indeed total: $D_0 = \ob{\catProbKern} \times \ob{\catProbKern}$; given $\memSpaceS{X}$ and $\memSpaceS{Y}$, $\memSpaceS{X} \tensor_0 \memSpaceS{Y} = \memSpaceS{X \cup Y}$; $\obE_0 = \memSpaceS{\empset}$.

    On morphisms, define $D_1 \subseteq \mor{\catProbKern} \times \mor{\catProbKern}$ to consist of precisely those probabilistic kernels that are parallelly composable; that is, a pair of morphisms $(f \colon \memSpaceS{X} \to \dist \memSpaceS{Y}, g \colon \memSpaceS{U} \to \dist \memSpaceS{V}) \in D_1$ iff $X \cap U = Y \cap V$. In this case, $f \tensor_1 g$ is $f \stPlus g$. The unit $\obE_1$ is $\idMor_{\empset}$.

    It remains to verify that $\lr{D_0, D_1}$ forms a subcategory of $\catProbKern \times \catProbKern$ -- denoted $\catD$,  and $\stPlus$ is a functor $\catD \to \catProbKern$.
    
    For the former, note that if $f_1 \stPlus g_1$ and $f_2 \stPlus g_2$ are both defined,  $f_1 \stThen f_2$ and $g_1 \stThen g_2$ are both defined, then $(f_1 \stThen f_2) \stPlus (g_1 \stThen g_2)$ is also defined by \eqref{frameAx:rev-exchange}.
    Also, the identity morphisms are present in $D_1$. Therefore $\lr{D_0, D_1}$ forms a subcategory $\catD$ of $\catProbKern$.

    For the latter, the functoriality spells out as: given arbitrary $(f_1, g_1), (f_2, g_2) \in \ob{\catD} = D_1 \times D_1$ that are sequentially composable, $(g_1 \after f_1) \stPlus (g_2 \after f_2) = (g_1 \stPlus g_2) \after (f_1 \stPlus f_2)$. This is guaranteed by \eqref{frameAx:rev-exchange}.

    Therefore $\lr{\catProbKern, \stPlus, \idMor_{\memSpaceS{\empset}}}$ is a PMC.
\end{proof}

\begin{restatable}{proposition}{PropCatFramePMC}
\label{prop:cat-frame-PMC}
    $\lrangle{\catKernel{\catWithVarName{\catC}{\obChoice}}, \stPlus, \idMor_{\emptyList}}$ is a partially monoidal category.
\end{restatable} 
\begin{proof}
    The proof is similar to that of \Cref{prop:prob-frame-PMC}, which is a concrete case of the current proposition.

    We define the partial monoids on the set of objects and morphisms as follows.

    For objects, its partial monoid structure $\lr{D_0, \tensor_0, \obE_0}$ is total: $D_0$ is simply $\ob{\catKernel{\catWithVarName{\catC}{\obChoice}}} \times \ob{\catKernel{\catWithVarName{\catC}{\obChoice}}}$, namely pairs of list representation of finite sets of variables; given $(L, K) \in D_0$, where $L = \setToList{X}$ and $K = \setToList{Y}$, $L \tensor_0 K = \setToList{X \cup Y}$ ; $\obE_0 = \emptyList$.

    For morphisms, its partial monoid structure $\lr{D_1, \tensor_1, \obE_1}$ is: $D_1 = \{ (f, g) \in \ob{\catKernel{\catWithVarName{\catC}{\obChoice}}} \times \ob{\catKernel{\catWithVarName{\catC}{\obChoice}}} \mid f \stPlus g \text{ is defined} \}$; given $(f, g) \in D_1$, $f \tensor_1 g = f \stPlus g$; $\obE_2 = \idMor_{\emptyList}$.

    With a bit abuse of notation, we shall simply denote both $\tensor_0$ and $\tensor_1$ simply by their corresponding operation $\stPlus$. 

    Next, that $\lr{D_0, D_1}$ forms a subcategory $\catD$ of $\catKernel{\catWithVarName{\catC}{\obChoice}} \times \catKernel{\catWithVarName{\catC}{\obChoice}}$ and $\stPlus \colon \catD \to \catKernel{\catWithVarName{\catC}{\obChoice}}$ is a functor follows from the frame conditions -- in particular \eqref{frameAx:rev-exchange}.
\end{proof}

\section{Omitted Proofs from \Cref{sec:CI}}
\label{appendix:sec-CI}
We fix a Markov category $\catC$ and a choice function $\obChoice \colon \setVar \to \ob{\catC}$.
\ThmCIEq*
\begin{proof}
    We prove the first point only. The proof follows immediately by spelling out the definition of $f \satisfy{\natVal} (\kp{\empset}{W}) \biThen ( (\kp{W}{X}) \mand (\kp{W}{Y}) )$. According to \Cref{def:DIBI-validity-sem-kernel}, it means there exist $\catWithVarName{\catC}{\obChoice}$-kernels $f_1, f_2$ such that $f = f_1 \stThen f_2$, $f_1 \satisfy{\natVal} (\kp{\empset}{W})$, and $f_2 \satisfy{\natVal} ( (\kp{W}{X}) \mand (\kp{W}{Y}) )$.

    $f_2 \satisfy{\natVal} ( (\kp{W}{X}) \mand (\kp{W}{Y}) )$ means there exist $\catWithVarName{\catC}{\obChoice}$-kernels $g_1, g_2$ such that $f_2 \stGeq g_1 \stPlus g_2$, $g_1 \satisfy{\natVal} (\kp{W}{X})$, and $g_2 \satisfy{\natVal} (\kp{W}{Y})$. By \Cref{def:kernel-composition}, \Cref{def:kernel-leq}, and $\natVal$, infer $g_1 = \input{thm-CI-eq-g1.tikz}$ and $g_2 = \input{thm-CI-eq-g2.tikz}$. 

    $f_1 \satisfy{\natVal} (\kp{\empset}{W})$ means, according to $\natVal$ and \Cref{def:kernel-composition}, $f_1$ is of the the form \input{thm-CI-eq-f1.tikz}. Given the form of $g_1$, $g_2$, and $f_2$, such $V$ -- which is already a subset of $W \cup X \cup Y \cup U$ by definition -- must satisfy $V \subseteq U$. 
    So far we have:
    \[
    f_2 \stGeq g_1 \stPlus g_2 = \input{thm-CI-eq-5.tikz} \text{ , therefore }
    f = \input{thm-CI-eq-7.tikz}
    \]
    where the two morphisms in the dotted boxes (represented only by dotted boxes later for simplicity) play the role of $g_1$ and $g_2$ in \Cref{def:proc-markov-CI}.\ref{item:def-markov-CI}, respectively. Delete $U$ in the output of $f$ by postcomposing $\discardMor_{U}$, one gets:
    \[
        \input{thm-CI-eq-12.tikz} = \input{thm-CI-eq-8.tikz} = \input{thm-CI-eq-9.tikz} = \input{thm-CI-eq-10.tikz}
    \]
    Compare the resulting diagram with \Cref{def:proc-markov-CI}.\ref{item:def-markov-CI}, it follows that $f$ displays Markov CI, namely $\markovCI{X}{Y}{W}$; in particular, here  \input{thm-CI-eq-11.tikz} plays the role of $s_{\obW}$ in \Cref{fig:display-CI-markov}.
\end{proof}

\begin{lemma}
\label{lem:two-CI-relation}
    In a Markov category $\catX$, superset CI implies Markov CI.
\end{lemma}
\begin{proof}
    Given an $\catX$-morphism $s \colon \tenUnit \to \obW \tensor \obX \tensor \obY \tensor \obU$ satisfying $\procCI{\obX}{\obY}{\obW}$ -- i.e., it can be decomposed as in \Cref{fig:display-CI-process}, we show that is also satisfies $\markovCI{\obX}{\obY}{\obW}$ -- i.e., it can be expressed as in \Cref{fig:display-CI-markov}.
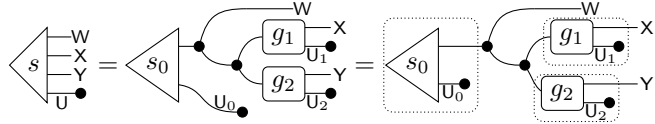
    \[
        \input{display-CI-markov-1.tikz} = \input{two-CI-relation-1.tikz} = \input{two-CI-relation-2.tikz}
    \]
    where the three diagrams in the dotted circles play the role of $s_{\obW}$, $g_{\obX}$, and $g_{\obY}$ in \Cref{fig:display-CI-markov}, respectively.
\end{proof}

\CorCIEq*
\begin{proof}
    Note that when $U = \empset$, both the superset CI statement $\procCI{X}{Y}{W}$ and the Markov CI statement $\markovCI{X}{Y}{W}$ reduce to the plain CI statement $\displayCI{X}{Y}{W}$. So the statement follows immediately from \Cref{thm:proc-DIBI-CI-eq} and \Cref{lem:two-CI-relation}.
\end{proof}

\PropDiagDibiCI*
\begin{proof}
    For the first part, we stick with the setting in \Cref{def:proc-markov-CI}, and further assume that $\catX$ has conditionals.
    
    Suppose $s$ satisfies $\markovCI{\obX}{\obY}{\obW}$, namely decomposition as in \Cref{fig:display-CI-markov} holds. Then, since $\catX$ has conditionals, there exist $g_{\obU} \colon \obW \tensor \obX \tensor \obY$ such that:
    \[
        \input{display-CI-process-1.tikz} = \input{markov-extend-superset-CI-eq-1.tikz} = \input{markov-extend-superset-CI-eq-2.tikz}
    \]
    The last diagram witness the extended superset CI $\extSupsetCI{\obX}{\obY}{\obW}$.

    Suppose $s$ satisfies $\extSupsetCI{\obX}{\obY}{\obW}$; i.e., $s$ can be decomposed as \eqref{eq:extend-superset-CI}. Then, deleting the $\obU$ part (where $\obV_0 \tensor \obV_1 \tensor \obV_2 \tensor \obV_3 = \obU$), we get:
    \[
        \input{display-CI-markov-1.tikz} = \input{markov-extend-superset-CI-eq-3.tikz} = \input{markov-extend-superset-CI-eq-4.tikz}
    \]
    This precisely says that $s$ satisfies Markov CI $\markovCI{\obX}{\obY}{\obW}$.
    
    For the second part of the statement, note that $\catC$ has conditionals implies that $\catWithVarName{\catC}{\obChoice}$ also has conditionals (\Cref{prop:markov-preserved}). Then it follows immediately from the first half of the current statement together with \Cref{thm:proc-DIBI-CI-eq}.
\end{proof}

\end{document}